\newif\iffull 
  \fulltrue

\documentclass[conference]{IEEEtran}

\usepackage{graphicx}
\usepackage{enumitem}
\usepackage[colorinlistoftodos]{todonotes}
\usepackage{ltl}
\usepackage{amsmath}
\usepackage{amsthm}
\usepackage{amsfonts}
\usepackage{stmaryrd}
\usepackage{dsfont}
\usepackage{xspace}
\usepackage{zlmtt}
\usepackage[svgnames,table]{xcolor}
\usepackage{todonotes}
\usepackage{colortbl}
\usepackage{amssymb}
\usepackage{cite}
\usepackage{float}
\usepackage{multirow}
\usepackage{booktabs}
\usepackage{xspace}
\usepackage{footmisc}
\usepackage{hhline}
\usepackage[colorlinks=true,urlcolor=black]{hyperref}
\usepackage{cleveref}
\usepackage{thmtools}
\usepackage{thm-restate}

\theoremstyle{plain}

\theoremstyle{definition}

\newtheoremstyle{boldnote}
{\topsep}
{\topsep}
{}
{}
{\bfseries}
{.}
{.5em}
{\thmname{#1}\thmnumber{ #2}\thmnote{ \textbf{(#3)}}}

\theoremstyle{boldnote}
\newtheorem{example}{Example}

\newcommand{\ie}[0]{i.e.,~}
\newcommand{\eg}[0]{e.g.,~}

\newcommand{\new}[1]{#1}

 \newcommand{\tightpar}[1]{{\smallskip\noindent\bf #1. }}

\newcommand{\hyperltl}{\textsc{HyperLTL}\xspace}
\newcommand{\hyperpctlStar}{\textsc{HyperPCTL*}\xspace}

\DeclareMathOperator{\Min}{\mathcal{M}\text{\textsc{in}}}
\DeclareMathOperator{\Max}{\mathcal{M}\text{\textsc{ax}}}
\DeclareMathOperator{\Ex}{\mathbb{E}}

\newcommand{\supp}{\ensuremath{\text{\textsl{supp}}}}

\newcommand{\valuation}[1]{\ensuremath{\llbracket #1 \rrbracket}}

\newcommand{\toolname}{\textsc{SVHyper}}

\begin{document}

\title{Statistical Verification of Quantitative Hyperproperties: Beyond Boolean Quantification}

 \author{
 	\IEEEauthorblockN{Amir M. Ahmadian and Hazem Torfah}
 	\IEEEauthorblockA{\textit{Chalmers University of Technology and University of Gothenburg}}
 }

\maketitle

\begin{abstract}
Formalisms for hyperproperties provide a solid foundation for studying the verification problem across classes of relational properties, such as those in information flow control (IFC). However, existing formalisms remain limited in expressiveness when it comes to capturing practical aspects of real-world systems. In particular, they do not adequately account for the quantitative nature of such systems.
In this paper, we address this gap by revisiting the specification and verification of hyperproperties from a quantitative, \emph{measure-based}, perspective. 

We introduce \emph{Quantitative Hyper-Logic} (QHL), which replaces qualitative trace quantifiers with measure-based ones and extends temporal predicates with richer quantitative expressions. 
We further study the verification problem from a statistical verification point of view, and develop  algorithms for the statistical verification of QHL specifications. For the introduced measure-based quantifiers, we particularly provide an analysis in terms of sample complexity and achievable statistical guarantees.
In particular, we show how statistical methods such as Hoeffding’s inequality and extreme value theory can be combined to develop statistical verification algorithms for nested measure-based quantifiers.
Our approach provides quantitative alternatives for where traditional verification methods become infeasible. We demonstrate both expressiveness and efficacy on benchmarks from quantitative IFC, comparing against qualitative methods.
\end{abstract}

\section{Introduction}\label{sec:intro}

Many system properties, particularly those concerning information flow control policies, define correctness in terms of relations over system executions. For example, prominent properties such as noninterference require that any two execution traces with identical low-security inputs produce identical low-security outputs~\cite{DBLP:conf/sp/GoguenM82a}. Such properties cannot be captured by standard trace properties, which reason about individual execution traces in isolation, and instead, are formalized as hyperproperties \cite{hyperprops}. Hyperproperties generalize trace properties from sets of traces to sets of sets of traces. Noninterference, for instance, is a hyperproperty that admits only those sets of traces in which every pair of traces satisfies the relation above.
Hyperproperties provide a unifying framework for studying relational properties in information flow control and beyond \cite{hyperprops,DBLP:conf/cav/HenzingerKKM23,hyperstl,DBLP:journals/tac/AnandMTZ24}. After their introduction, several formalisms have been proposed for specifying hyperproperties, along with a range of formal verification techniques \cite{hyperprops,hyperltl,hyperstl,DBLP:conf/concur/GutsfeldMO20,hyperpctl,secltl}. However, despite offering a solid theoretical foundation, existing approaches still struggle to capture the complexity of real-world systems. In particular, they often lack the expressive power needed to model quantitative aspects that are inherent in practical settings.

In this paper, we initiate the study of quantitative hyperproperties from a measure-theoretic point of view. 
We introduce \emph{Quantitative Hyper-Logic} (QHL), a logic for expressing quantitative hyperproperties through measure-based quantification and  expressions. 
QHL moves beyond the  standard approach to hyperproperty specification, which extends trace logics with explicit trace quantifiers that allow for explicit referencing of traces\footnote{The focus of this paper is on linear-time logics for hyperproperties, in contrast to branching-time. We leave a full study of the latter for future work.}. 
Instead, QHL replaces such quantification with measure-based quantifiers, which define measures over sets of traces. 
For example, in the case of noninterference, considered in its qualitative form as described above, we can express the property in \textsc{HyperLTL} \cite{hyperltl} as follows: 
$$\forall \pi \forall \pi'.~ \LTLsquare (\bigwedge_{a\in LI} a_\pi = a_{\pi'}) \rightarrow \LTLsquare(\bigwedge_{a\in LO} a_\pi = a_{\pi'}),$$
where $\mathit{LI}$ and $\mathit{LO}$ define the sets of (binary) low inputs and outputs, respectively.  
Such a property is, however, too strict in practice. Programs as simple as a password checker already violate this property by providing information whether an input phrase matches a password. 
A more adequate formulation would thus be a quantitative one that bounds the amount of information leaked \cite{qif}. This, however, cannot be accurately described in qualitative logics like \textsc{HyperLTL}. 
Examples of such formulations include quantitative generalizations of noninterference such as those
based on notions of (Bayes) vulnerability \cite{qif,smith2015recent}, given by $V(X) = \max_{x} \mathbb{P}[X=x]$, which defines the worst-case probability that an adversary can guess the value of secret~$X$. 
\new{In QHL, assuming that the adversary's candidate guesses range over the space of secrets, we can formulate this as the measure-based specification:}
$$\Max \pi \ \Ex \pi'.~ \mathds{1}(\pi =_{H} \pi') \leq \lambda.$$
\new{
The formula states that for any trace $\pi$, which contributes a candidate guess by the adversary, we want the expected value of encountering a trace $\pi'$ whose secret matches that guess (expressed as $\pi =_{H} \pi'$) to be below a threshold $\lambda$, and we want this to hold in the worst case ($\Max$). 
The expected value corresponds to the probability that the guess is correct.
}
As we show in the paper, QHL is a very expressive logic and through the nesting of different measure-based quantifiers will allow us to express a rich variety of quantitative hyperproperties.

Previous work on quantitative logics for hyperproperties has  focused on extending probabilistic temporal logics with suitable forms of quantification. These approaches enable reasoning about distributions over system behaviors in purely probabilistic settings \cite{hyperpctl}, as well as quantification over strategies in models that additionally incorporate nondeterminism~\cite{phl,DBLP:conf/atva/AbrahamBBD20}. In such frameworks, however, the quantitative aspect stems from the probabilistic semantics of the underlying logic, while the quantification itself follows the same qualitative principles as in standard hyper-logics. 
Some approaches go a step further,  introducing probability measures over sets of traces \cite{wang2021hyperpctl_star}. These remain, however, limited to probabilistic measures applied to temporal predicates. In contrast, we move beyond purely probabilistic notions by considering more general classes of measures over sets of traces. We will especially study quantification in terms of expectations, minima, and maxima over quantitative expressions, and show that such measures naturally capture a richer class of quantitative hyperproperties.

The study of measure-based quantification also raises fundamental questions regarding verifiability. We investigate the verifiability in a black-box setting and present a set of approaches for the statistical   verification of QHL specifications. 
Different classes of measures require the use of distinct theories and tools from statistical theory. 
For instance, measures defined in terms of expectations can be addressed using sampling techniques, together with results such as Hoeffding's inequality \cite{hoeffding1963probability}. Measures such as maxima and minima, however, require the use of results from extreme value theory \cite{coles2001EVT}.
We develop statistical verification methods tailored to the different classes of measures and study them in terms of statistical guarantees and sample complexity. In particular, we investigate the impact of \emph{nesting and alternations} between different measure-based quantifiers, and how to combine the corresponding statistical techniques while preserving formal guarantees. Using a prototypical implementation, we further demonstrate the applicability of the proposed methods on use cases from the domain of quantitative information flow control, and compare them with existing qualitative and probabilistic approaches.
 
We summarize our contributions as follows: 
\begin{itemize}
    \item We initiate the study of measure-based hyper-logics as a richer formalism for defining quantitative hyperproperties.
    \item We introduce QHL, a logic that enables reasoning about quantitative hyperproperties in terms of measures over relations on traces. 
    \item We develop statistical verification methods based on adaptations of core statistical theories,  particularly studying the adaptation in the presence of  nested alternating quantifiers. 
    \item Lastly, we present experimental results demonstrating the applicability of the proposed methods. 
\end{itemize}

\tightpar{Outline} 
\new{In the next section we survey related work on hyperproperties, quantitative information flow, and probabilistic approaches, and position our work within this landscape.}
In \Cref{sec:qhl}, we present the new logic QHL, and compare the expressivity of the logic to previous qualitative and probabilistic hyper-logics in \Cref{sec:relation-to-logics}. In \Cref{sec:statistical_verification}, we present a set of  statistical verification methods for verifying QHL. 
Our experimental evaluation is given in \Cref{sec:evaluation}. In \Cref{sec:conclusion}, we briefly conclude with pointers to future directions.

\section{Related Work}\label{sec:related_works}
We position our contribution relative to the state of the art on hyperproperties and quantitative information flow (QIF).

 \tightpar{Hyperproperty Formalisms}
 Since the introduction of hyperproperties \cite{hyperprops}, several formalisms have been proposed for their specification.
 In their seminal work, Clarkson and Schneider already provide a logical characterization of hyperproperties using second-order logic. However, this characterization was generally not verifiable. 
 For some time, logical approaches to hyperproperties were restricted to reductions into standard temporal logics via self-composition techniques \cite{DBLP:conf/csfw/HuismanWS06, DBLP:journals/mscs/BartheDR11}.
 The landscape shifted with the introduction of logics like \textsc{HyperLTL}~\cite{hyperltl}. \textsc{HyperLTL} extends LTL (linear-time temporal logic \cite{ltl}) with universal and existential quantifiers over traces, enabling reasoning over tuples of executions. 
 Similar extensions have been developed for logics over continuous signals, leading to logics like \textsc{HyperSTL} \cite{hyperstl}. Hyperproperties have also been studied in branching-time settings, most notably with logics like \textsc{HyperCTL*}. 
Later work introduced more expressive logics such as \textsc{HyperPDL} \cite{DBLP:conf/concur/GutsfeldMO20}, enabling reasoning about $\omega$-regular properties beyond those in \textsc{HyperLTL} and \textsc{HyperCTL*}. 
All these formalisms allow us to reason about qualitative hyperproperties, 
inheriting the qualitative semantics of their underlying logics. 
Quantitative reasoning in these logics is limited to defining cardinality constraints by explicitly quantifying over a fixed number of traces \cite{DBLP:conf/cav/FinkbeinerHT18,sahai2020verification}.

In recent years, quantitative formalisms for specifying hyperproperties have largely focused on probabilistic settings. 
For example, \textsc{HyperPCTL} \cite{hyperpctl} was introduced to define properties of purely probabilistic systems by extending the probabilistic temporal logic PCTL with quantification over states. 
Logics that further allow reasoning about systems that combine both probabilistic and nondeterministic behavior were proposed by extending \textsc{HyperPCTL} with quantifiers over strategies \cite{DBLP:conf/atva/AbrahamBBD20}, and with logics like \textsc{PHL} \cite{phl}, which extends \textsc{HyperCTL*} with  a probabilistic operator and quantifiers over strategies. 
Further variants that allow for the nesting of temporal and probabilistic operators have also been proposed, like \textsc{HyperPCTL*} \cite{wang2021hyperpctl_star}. 
With the exception of \textsc{HyperPCTL*}, existing probabilistic hyper-logics do not support reasoning in terms of measures over relations between traces. 
Instead, their quantitative nature derives from the probabilistic semantics of the underlying logic, while the quantification itself follows the same qualitative principles as in standard hyper-logics. 
\textsc{HyperPCTL*} goes beyond this, to allow probability measures over sets of traces, but remains restricted to probabilistic measures applied
to temporal predicates. 
In contrast, we expand the study to a more general perspective by introducing measure-based quantifiers and quantitative expressions that capture a richer class of quantitative hyperproperties.

The development of these formalisms has been accompanied by efforts to study their verification problems \cite{hyperalg,DBLP:conf/cav/FinkbeinerHT18,phl,DBLP:conf/csfw/BonakdarpourF18}. 
Statistical verification methods have especially been proposed for probabilistic hyper-logics, with existing work primarily focusing on direct adaptations of statistical model checking (SMC) techniques, such as sequential probability ratio tests \cite{10.1145/3358232,wang2021hyperpctl_star}.
Research on statistical methods for quantitative hyperproperties remains in its infancy. 
In this work, we take a step in this direction by initiating a systematic study of statistical verification for richer, measure-based formalisms of quantitative hyperproperties, and exploring how classical statistical theories can be adapted to this setting.

\tightpar{Quantitative Information Flow}%
\new{
In a largely independent line of work in QIF, a rich body of leakage measures and analyses has been developed. 
Beyond the notion of vulnerability \cite{qif}, the g-leakage framework \cite{alvim2012measuring,alvim2020science} parameterizes leakage by a gain function that rewards an adversary's guess according to the actual secret. This unifies Bayes vulnerability, multi-try attacks, and guessing entropy as instances of a single quantity computed over the hyper-distribution induced by a program's input–output channel.
As we show in Section \ref{sec:examples}, this family of measures is directly expressible in QHL. %
}

\new{
Hyper-logics and QIF have been related once before, by classifying QIF verification problems as a safety and liveness hyperproperties \cite{yasuoka2014quantitative}. 
Our measure-based perspective is complementary as it provides a general-purpose  logic: rather than defining each leakage measure separately, QIF-style measures are written as formulas in a common syntax, and are therefore verified by a single algorithm.
QHL thus fills a gap between the discrete, exactly verified quantitative-hyperproperty line \cite{DBLP:conf/cav/FinkbeinerHT18,sahai2020verification} and the measure-specific QIF tradition \cite{alvim2020science}. 
We do not aim to replace specialized leakage measures or their accompanying theories, but to provide the specification and verification machinery needed once such a measure has been chosen and must be nested, combined, or compared across several sub-populations of system traces.
}

\new{
Statistical estimation of leakage has also been studied, across several methodological families: purely black-box estimators computed from repeated executions with statistical confidence guarantees \cite{chothia2013tool,chothia2014leakwatch,chatzikokolakis2010statistical}, hybrid methods that combine precise and statistical analysis over program components to reduce sample counts \cite{kawamoto2016hybrid}, symbolic approaches that rely on symbolic sampling to yield provably upper and lower bounds \cite{malacaria2018symbolic}, and fuzzing-based tools that estimate leakage dynamically and steer the search toward leaking executions \cite{blackwell2025nifuzz}.
All of these estimate a single, fixed leakage quantity over one secret–observable channel, and none addresses the compositional verification of nested, alternating measure-based quantifiers that is our concern here. 
We further note that our approach is agnostic to the sampling approach used at each quantifier, so the more sample-efficient techniques above could in principle replace the plain Monte Carlo sampling used in our prototypical engine.
}

\tightpar{Beyond QIF}%
Finally, we note that while the study of hyperproperties has predominantly focused on security, its scope extends well beyond this domain to include areas like robustness \cite{DBLP:conf/atva/SeshiaDDFGKSVY18}, fairness \cite{DBLP:conf/cav/HenzingerKKM23}, and more. 
Our measure-based perspective extends to these settings and will provide richer formalisms and methods for specifying and verifying such properties.

\section{Quantitative Hyper-Logic}\label{sec:qhl}

In this section, we introduce the syntax and semantics of QHL. We start with a few examples, from the domain of quantitative information flow, that illustrate the expressivity of our logic in capturing key quantitative hyperproperties. 

\subsection{Example QHL Specifications}
\label{sec:examples}

\begin{example}[Authentication System Vulnerability]
We expand on the notion of \new{Bayes} vulnerability discussed earlier ($V(X) = \max_{x} \mathbb{P}[X=x]$). 
Measures based on vulnerability are crucial for evaluating the security of authentication systems against optimal, one-shot attacks such as targeted password spraying. 
The distribution of passwords selected by human users is in general not uniform. 
Analysis of real-world breach datasets demonstrates that password frequencies follow Zipf's Law \cite{wang2017zipf}. 
Driven by the cognitive bias of simplicity and the principle of least effort, users create predictable, non-random passwords. 
This means that a small number of very common passwords (\eg ``123456'', ``password'') are used more frequently, while the vast majority of passwords are less frequent.

\new{
Let \texttt{authSuccess} be an indicator function over pairs of traces, which evaluates to $1$ if the attacker's guess in $\pi_1$ matches the user password in $\pi_2$, and $0$ otherwise. 
A QHL formula can now be defined to reason about the vulnerability of a password-based authentication system as follows:
$$\Max \pi_1 \ \Ex \pi_2. \ \texttt{authSuccess}(\pi_1, \pi_2).$$
The two traces play different roles as $\pi_1$ contributes only the adversary's candidate guess, and $\pi_2$ only the user's password. 
The inner expectation therefore assumes the guess is fixed and averages the indicator function over the distribution of passwords, yielding the probability that this particular guess is correct. 
The worst-case probability is then captured by maximizing over the adversary's guesses, which yields $V(X)$ provided the guesses occurring in the trace space range over the space of passwords.
}

\new{
The indicator \texttt{authSuccess} acts as the identity gain function~\cite{alvim2012measuring}, rewarding the adversary only for a guess that matches the password exactly, so that the formula above computes  Bayes vulnerability $V(X)$, the worst-case probability of success under a single guess. 
A similar formula can be used to capture the wider family of gain-function-based measures of the g-leakage framework~\cite{alvim2012measuring}. 
Replacing the indicator by a gain function $g(\pi_1, \pi_2)$, which rewards the guess carried by $\pi_1$ according to the password carried by $\pi_2$, we obtain the specification:
$$\Max \pi_1 \ \Ex \pi_2. \ g(\pi_1, \pi_2)$$
computes the g-vulnerability $V_g(X) = \max_{w} \sum_{x} \mathbb{P}[X = x] \cdot g(w, x)$, where $w$ is the attacker guess and $X$ is the password.
The inner $\Ex$ is the expected gain of a fixed guess, and the outer $\Max$ selects the optimal one over the guesses. 
Bayes vulnerability is the instance obtained by taking $g$ to be the identity gain function~\cite{alvim2012measuring}.
}

\new{
The specifications above quantify the vulnerability of the secret before anything about a particular execution has been observed. 
Systems generally also expose an adversary-observable output $Y$, and the corresponding posterior measures ask how vulnerable the secret becomes once an output value has been observed by the attacker. 
The posterior g-vulnerability is defined as $\sum_y \mathbb{P}[Y = y] \max_w \sum_x \mathbb{P}[X = x \mid Y = y]\, g(w,x)$, which differs from the prior measure in two respects: the secret is averaged over the conditional distribution $\mathbb{P}[X = x \mid Y = y]$ rather than over its full distribution, and this evaluation is itself averaged over the observations.
}

\new{
QHL captures the latter with a further nested expectation quantifier, and the former with a \emph{relational trace constraint}. 
A quantifier in QHL may be restricted by a relation over traces, written $\pi_{\sim \psi}$, which limits the measure space of that quantifier to the traces in the relation $\psi$. 
Such a restriction is precisely a conditioning of the system's measure, and therefore supplies the conditional distribution that the posterior measures require. 
Writing $R_{obs}(\pi_0, \pi_2)$ for the relation that holds when two traces agree on the observable output, we obtain the specification:
$$\Ex \pi_0 \ \Max \pi_1 \ \Ex \pi_{2 \sim R_{obs}(\pi_0,\pi_2)}. \ g(\pi_1, \pi_2)$$
computes the posterior g-vulnerability.
The trace $\pi_0$ contributes only an observation, the innermost $\Ex$ averages the gain over the passwords that are consistent with that observation, the $\Max$ selects the adversary's best guess for it, and the outer expectation averages over observations weighted by their probability.
The alternation of extremal and expectation quantifiers in this specification is exactly the structure that the statistical verification algorithms of Section \ref{sec:statistical_verification} are designed to handle.
}
A verification problem can be defined by introducing a bound $\lambda$ and verifying whether evaluating the above QHL specification results in  a value below $\lambda$ \footnote{We note that one could also introduce a syntactic layer to QHL explicitly specifying comparisons between QHL formulas. We leave this out for simplicity and maintain a purely measure-based semantics for QHL. Our verification algorithms are not impacted by this choice.}. 
\end{example}

\begin{example}[Power Side-Channel Attacks]
In hardware security, cryptographic devices (such as smartcards) inadvertently leak information through physical side-channels, most notably power consumption. 
Attackers can exploit this via Differential Power Analysis (DPA) \cite{kocher1999dpa} to recover secret keys. 
We want to measure the absolute maximum expected leakage capacity: given an identical attacker-controlled input (the plaintext), what is the worst-case expected observable distance in continuous power consumption (in millivolts) caused by the variation of the high-security input (the secret key)?

Using a quantitative distance metric \texttt{powerDist}, which calculates the absolute difference in power consumption between the two traces, based on the value of the secret key, we can define the following specification: 
$$\Max \pi_1 \ \Ex \pi_{2 \sim {R_{\texttt{samePT}}(\pi_1, \pi_2)}} . \ \texttt{powerDist}(\pi_1, \pi_2),$$
where $R_{\texttt{samePT}}(\pi_1, \pi_2)$ is a relational trace constraint that restricts the measure space of the inner expectation to traces where the attacker's chosen plaintext in $\pi_2$ matches the candidate plaintext explored in $\pi_1$.  
The secret key remains unconstrained. 
\end{example}

\begin{example}[Probabilistic Noninterference]
Probabilistic noninterference \cite{gray1992prop_noninter,sabelfeld2000probabilistic} is a quantitative variant of noninterference that states that for two traces with identical low inputs we cannot distinguish the probability distributions of low outputs. 

In QHL we can define such a property as follows. 
We define the relational constraint $R_{LI}(\pi_1, \pi_2)$, which constrains the two baseline traces to have identical low initial inputs, meaning input-wise they differ only in their high inputs. 
Probabilistic noninterference can then be given by the specification:
\begin{align*}
    \Max\ \pi_1 \ & \Max \pi_{2 \sim {R_\mathit{LI}(\pi_1, \pi_2)}}.~ \\
    \Big| & \Ex {\pi'_1}_{\sim R_\mathit{LI}(\pi_1, \pi'_1)}. \texttt{termO}(\pi_1,\pi'_1) - \\
    & \Ex {\pi'_2}_{\sim R_\mathit{LI}(\pi_2, \pi'_2)}. \texttt{termO}(\pi_2,\pi'_2)\Big|,
\end{align*}
where \texttt{termO} is an indicator function on output equality. 
The baseline traces with identical low inputs are defined using the variables $\pi_1$ and $\pi_2$. 
These can have different valuations of the high inputs. 
For each trace $\pi_i$, for $i\in \{1,2\}$, we use the inner expressions,  $\Ex {\pi'_i}_{\sim R_{LI}(\pi_i, \pi'_i)} \texttt{termO}(\pi_i,\pi'_i)$, to capture the probability of observing an output for a given low input. 
The difference between the two expected value expressions defines the difference in output probabilities for the observed inputs, and the $\Max$ quantifier returns the worst-case difference.
\end{example}

\subsection{Notation}
In our work, we assume an underlying system, whose observable behavior can be modeled by a Markov chain. 
A Markov chain is modeled as a  sequence of random variables $X_1,X_2, \dots$ satisfying the Markov property, i.e., for any $i\ge 0$, it holds that $Pr(X_{i+1}=v_{i+1}\mid X_i=v_i) = Pr(X_{i+1} = v_{i+1} \mid  X_0=v_0, X_1=v_1, \dots, X_i=v_i)$. 
We define a trace $t = v_1v_2 \dots$ as an infinite sequence of valuations of the random variable, representing a single execution of the system, such that for any $i$ it holds that $Pr(X_{i+1} = v_{i+1} \mid X_i = v_i) >0$. 
We define  $T$ to be the space of all such infinite observation sequences induced by the Markov chain.
We define a finite prefix of a trace $t$ up to step $k$ as $t[..k]$.
To reason about the stochastic behavior of the system, we assume $T$ is a measurable space equipped with a probability measure $\mu$. 
The valid executions of the system are captured by the support of the measure, $\supp(\mu) \subseteq T$, which denotes the set of traces whose finite prefixes have a non-zero probability.

To formally specify properties over these traces, we assume a countable set of trace variables $\textit{Var} = \{\pi_1, \pi_2, \dots\}$, which range over $T$, and a set of quantitative functions $f: T^k \to \mathbb{R}$, which map $k$-tuples of traces to real-valued quantitative metrics (\eg execution time, energy consumed, or distance).

\subsection{Syntax}
The syntax of QHL is defined over three components: 
(1) \emph{relational trace formulas} ($\psi$) which define structural, observational, or temporal constraints between traces, 
(2) \emph{trace measures} ($\varphi$) which evaluate quantitative metrics over a $k$-tuple of traces, and 
(3) \emph{set measures} ($\Phi$) which aggregate those trace measures across the system's probability measure to quantify the system's overall behavior.

The grammar of QHL is defined as follows:
\begin{align*}
	\psi ::= \ &\text{true} \mid R(\pi_1, \dots, \pi_n) \mid \psi \land \psi \mid \neg \psi \\
	\varphi ::= \ &c \mid f(\pi_1, \dots, \pi_k) \\
	\Phi ::= \ &\varphi \mid \Max \pi_{\sim \psi} . \ \Phi \mid \Min \pi_{\sim \psi} . \ \Phi \mid \\ &\Ex \pi_{\sim \psi} . \ \Phi \mid |\Phi| \mid \Phi \oplus \Phi 
\end{align*}
where $R$ is an arbitrary $n$-ary relational predicate over traces, $c \in \mathbb{R}$ is a constant, $f$ is a $k$-ary quantitative function, and $\oplus \in \{+,-\}$.
The notation $\pi_{\sim \psi}$ indicates that valuations of the trace variable $\pi$ belong to the space defined by the relation $\psi$.

\subsection{Semantics}
QHL formulas are interpreted over a probability measure $\mu$ and a trace assignment map $\Pi\colon \mathit{Var} \to T$, which maps each trace variable to a specific execution trace from $T$. 

We first define the satisfaction relation $\Pi \models \psi$ for relational trace formulas:
\begin{align*}
	\Pi &\models \text{true}  &&\iff \text{true} \\
	\Pi &\models R(\pi_1, \dots, \pi_n) &&\iff (\Pi(\pi_1), \dots, \Pi(\pi_n)) \in \mathcal{I}(R) \\
	\Pi &\models \psi_1 \land \psi_2 &&\iff \Pi \models \psi_1 \text{ and } \Pi \models \psi_2 \\
	\Pi &\models \neg \psi &&\iff \Pi \not\models \psi
\end{align*}
where $\mathcal{I}(R)$ is the semantic interpretation of the relation $R$ over $T^n$. 
This interpretation formally maps the relation $R$ to the explicit mathematical set of $n$-tuples of traces that satisfy the condition.
	For example, we can model an input-equivalence restriction up to step $k$ by defining a binary relation symbol $R_{in}^k$. Its formal interpretation is defined as the set of all trace pairs that share identical inputs:
	\begin{align*}
		\mathcal{I}(R_{in}^k) = \{ (t_1, t_2) \in T \times T \mid t_1[..k] =_{in} t_2[..k] \}
	\end{align*}
	Thus, the semantic rule $\Pi \models R_{in}^k(\pi_1, \pi_2)$ dictates that the relational formula evaluates to true if and only if the specific execution traces currently assigned to variables $\pi_1$ and $\pi_2$ by $\Pi$ belong to this equivalence set.

Given an existing trace assignment $\Pi$ and a relational constraint~$\psi$, the set of valid traces that satisfy the constraint for a new trace variable $\pi \in \mathit{Var}$ is defined as the acceptance set:
\begin{align*}
	Sat(\pi, \psi, \Pi) = \{ t \in T \mid \Pi[\pi \mapsto t] \models \psi \}
\end{align*}

The restricted measure $\mu_{\psi, \pi}^\Pi$ is the conditional probability of a new trace $\pi$, given the constraints imposed by previously sampled traces. 
For any measurable subset of traces $A \subseteq T$, it is defined as:
\begin{align*}
	\mu_{\psi, \pi}^\Pi(A) = \frac{\mu(A \cap Sat(\pi, \psi, \Pi))}{\mu(Sat(\pi, \psi, \Pi))}
\end{align*}
where the numerator isolates the probability of observing a trace that both belongs to the target set $A$ and satisfies the relational constraint $\psi$, and the denominator then divides this by the total probability of the constraint being satisfied in the underlying system. 
This operation effectively discards all invalid traces from $T$ and normalizes the probability mass of the valid traces.

We define the quantitative valuation function $\valuation{\cdot}^\mu_\Pi$, which maps a QHL formula to a real number evaluated under the system's probability measure $\mu$ and the assignments in $\Pi$:
\begin{align*}
	\valuation{ c }^\mu_\Pi &= c \\
	\valuation{ f(\pi_1, \dots, \pi_k) }^\mu_\Pi &= f(\Pi(\pi_1), \dots, \Pi(\pi_k)) \\
	\valuation{ \Max \pi_{\sim \psi} . \ \Phi }^\mu_\Pi &= \sup_{t \in \text{supp}(\mu_{\psi, \pi}^\Pi)} \valuation{ \Phi }^\mu_{\Pi[\pi \mapsto t]} \\
	\valuation{ \Min \pi_{\sim \psi} . \ \Phi }^\mu_\Pi &= \inf_{t \in \text{supp}(\mu_{\psi, \pi}^\Pi)} \valuation{ \Phi }^\mu_{\Pi[\pi \mapsto t]} \\
	\valuation{ \Ex \pi_{\sim \psi} . \ \Phi }^\mu_\Pi &= \int_{T} \valuation{ \Phi }^\mu_{\Pi[\pi \mapsto t]} \, d\mu_{\psi, \pi}^\Pi(t)
\end{align*}

The set measure quantifiers $\Max$, $\Min$, and $\Ex$ dictate how the quantitative outcomes of the inner trace measures are aggregated into a system-wide metric. 
By explicitly incorporating the conditional measure $\mu_{\psi, \pi}^\Pi$ into the quantifier semantics, QHL provides a unified mechanism for evaluating hyperproperties under specific observational or structural constraints.

The $\Max$ and $\Min$ operators search for the supremum and infimum outcomes, respectively. 
By strictly bounding this search to the support of the conditional probability measure ($\text{supp}(\mu_{\psi, \pi}^\Pi)$), the logic ensures that it only evaluates executions that satisfy the relational constraint $\psi$ and have a valid, non-zero chance of occurring. 
Similarly, the $\Ex$ operator evaluates the expected, average-case behavior of the system within the restricted domain. 
It achieves this by computing the integral over the trace space $T$ (since $T$ is uncountably infinite), calculating a weighted average by multiplying the quantitative valuation of each individual trace ($\valuation{\Phi}_{\Pi[\pi \mapsto t]}$) by its conditional probability mass ($d\mu_{\psi, \pi}^\Pi(t)$).

\section{Relation to Existing Logics}\label{sec:relation-to-logics}
In this section, we present how existing qualitative and probabilistic hyper-logics can be expressed in QHL. 
These logics typically evaluate properties over a strict Boolean domain, either implying the satisfiability of a formula, or computing the probability of a Boolean event. 
To formally establish a relation between QHL and these logics, we must first bridge the Boolean-to-quantitative gap.

We embed standard Boolean outcomes into QHL's quantitative semantics by defining a quantitative trace measure as an indicator function. 
Let $\varphi_{bool}$ be a standard quantifier-free Boolean trace (or path) formula. 
The indicator function $\mathbb{I}_{\varphi_{bool}} : T^n \to \{0, 1\}$ maps a tuple of traces to $1$ if they satisfy the property, and $0$ otherwise:
\begin{align*}
	\valuation{ \mathbb{I}_{\varphi_{bool}}(\pi_1, \dots, \pi_n) }^\mu_\Pi = 
	\begin{cases} 
		1 & \text{if } \Pi \models \varphi_{bool}(\pi_1, \dots, \pi_n) \\
		0 & \text{if } \Pi \not\models \varphi_{bool}(\pi_1, \dots, \pi_n)
	\end{cases}
\end{align*}

This indicator function is essential in allowing QHL to natively evaluate and aggregate Boolean satisfiability within its real-valued framework.

\subsection{Relation to Qualitative Logics}\label{sec:relation_to_qualitative_logics}

Existing qualitative hyper-logics evaluate properties by asking whether \emph{all} traces satisfy or \emph{at least one} trace satisfies a relation. 
In QHL, we subsume these qualitative questions by applying extreme set measures ($\Min$ and $\Max$) to the indicator function, effectively searching for worst-case and best-case Boolean satisfiability across the system's measure space.

To demonstrate this relation, we use \hyperltl \cite{hyperltl} as a baseline.
We define a structural translation function $\mathcal{T}$ that recursively maps a \hyperltl formula $\psi$ into a QHL formula. 
Let $\varphi$ be a quantifier-free inner LTL formula. 
The translation $\mathcal{T}$ is defined as follows:
\begin{align*}
	\mathcal{T}(\forall \pi. \ \psi) &= \Min \pi. \ \mathcal{T}(\psi) \\
	\mathcal{T}(\exists \pi. \ \psi) &= \Max \pi. \ \mathcal{T}(\psi) \\
	\mathcal{T}(\varphi) &= \mathbb{I}_{\varphi}
\end{align*}
Under this structural translation, a system satisfies a \hyperltl formula $\psi$ if and only if the quantitative valuation of its QHL translation evaluates to $1$:
\begin{align*}
	\Pi \models \psi \iff \valuation{ \mathcal{T}(\psi) }^\mu_\Pi = 1
\end{align*}

In this mapping, universal quantification ($\forall$) maps to finding the minimum ($\Min$), which intuitively means that to prove that all traces satisfy a Boolean property, the absolute worst-case evaluation across the trace space must still evaluate to $1$. 
Conversely, existential quantification ($\exists$) maps to finding the maximum ($\Max$). 
To prove that at least one trace satisfies a property, the best-case evaluation across the trace space must be $1$, effectively indicating that a witness exists.

\begin{example}
	Observational determinism states that given the same public inputs, the public observations must be identical across all execution traces:
	\begin{align*}
		\forall \pi_1 \ \forall \pi_2 . \ (\pi_1 =_{L,in} \pi_2) \rightarrow \Big( obs_L(\pi_1) = obs_L(\pi_2) \Big)
	\end{align*}
	
	By applying our structural translation $\mathcal{T}$, the nested universal quantifiers are translated into nested $\Min$ quantifiers. 
	Let $\mathbb{I}_{obs\_det}$ be the Boolean indicator function that returns $1$ if the two input traces have different public inputs \emph{or} if they have the same public inputs and their public observations match perfectly, and $0$ otherwise. 
	
	In QHL, ensuring that this property holds translates to searching for the worst-case pair of arbitrary traces and ensuring their evaluation is $1$:
	\begin{align*}
		\valuation{ \Min \pi_1 \ \Min \pi_2 . \ \mathbb{I}_{obs\_det}(\pi_1, \pi_2) }^\mu_\Pi = 1
	\end{align*}
\end{example}

\subsection{Relation to Probabilistic Logics}\label{sec:relation_to_probabilistic_logics}

Beyond qualitative extremes, QHL's expectation quantifier ($\Ex$) naturally extends its expressiveness into the domain of probabilistic hyper-logics. 
In probability theory, the expected value of an indicator function over a measurable space is equivalent to the probability measure of the underlying event.
We exploit this fundamental equivalence to embed Boolean probabilistic queries directly into QHL's quantitative semantics.
By evaluating the indicator function $\mathbb{I}_{\varphi_{bool}}$ under QHL's expectation quantifier ($\Ex$), the integral over the probability measure $\mu$ computes the exact probability mass of the satisfying traces.

To demonstrate this relation, we use \hyperpctlStar \cite{wang2021hyperpctl_star} as a baseline.
\hyperpctlStar was introduced to overcome the limitations of early probabilistic hyper-logics by relying on quantification over execution paths ($\pi$) instead of state quantifiers ($\forall s, \exists s$). 
This path-centric semantics closely resembles QHL's linear-time, trace-based semantics, which allows us to draw a direct connection between the \emph{linear-time} portion of \hyperpctlStar and QHL.

Thus, the \hyperpctlStar probability operator over a tuple of paths $\mathbb{P}^{(\pi_1, \dots, \pi_n)}$ translates to a sequence of QHL expectation quantifiers:
\begin{align*}
	\Pi \models \mathbb{P}^{(\pi_1, \dots, \pi_n)} \varphi_{bool} \iff \valuation{ \Ex \pi_1 \dots \Ex \pi_n . \ \mathbb{I}_{\varphi_{bool}}(\pi_1, \dots, \pi_n) }^\mu_\Pi
\end{align*}

One of the primary contributions of \hyperpctlStar is its ability to express the \emph{arithmetic of probabilistic quantifications} by applying arbitrary elementary functions over probabilities. 
While \hyperpctlStar supports a broad range of mathematical operations (such as products or entropy calculations), QHL only captures their \emph{linear} subset (\eg comparing the sum or difference of probabilities across different execution paths). 
Because QHL natively treats expected values as real numbers, this linear probability arithmetic translates cleanly to QHL's standard arithmetic syntax ($\oplus$).

\begin{example}
	Consider a \hyperpctlStar property stating that the satisfaction probability of reaching a property $a$ on an arbitrary path $\pi_1$ is greater, by at least a constant $c$, than the probability of reaching $b$ on a different arbitrary path $\pi_2$:
	\begin{align*}
		\mathbb{P}^{\pi_1}(\LTLdiamond a^{\pi_1}) - \mathbb{P}^{\pi_2}(\LTLdiamond b^{\pi_2}) > c
	\end{align*}
	
	Let $\mathbb{I}_{\LTLdiamond a}$ and $\mathbb{I}_{\LTLdiamond b}$ be indicator functions that return $1$ if the input trace eventually reaches a state with property $a$ and $b$, respectively. 
	By framing the inequality as a quantitative evaluation, this can be expressed in QHL as:
	\begin{align*}
		\Big( \Ex \pi_1 . \ \mathbb{I}_{\LTLdiamond a}(\pi_1) \Big) - \Big( \Ex \pi_2 . \ \mathbb{I}_{\LTLdiamond b}(\pi_2) \Big)
	\end{align*}
	and the result of this formula can then be checked to be $ > c$.
\end{example}

It is worth noting that QHL's relationship to \hyperpctlStar is ultimately one of \emph{partial subsumption}. 
\hyperpctlStar has branching-time semantics, allowing probability operators to be nested \emph{inside} temporal operators. 
For example, a property such as $\mathbb{P}^{\pi_1}(\LTLdiamond (\mathbb{P}^{\pi_2}(a^{\pi_2}) > c_2)) > c_1$ evaluates a path $\pi_1$ until some unknown future step ($\LTLdiamond$), captures that exact mid-execution state, and spawns a new path $\pi_2$ branching out from that moment to evaluate its probability. 
\new{As noted earlier, QHL is strictly linear-time and trace-centric.}
While QHL supports \emph{structural nesting} (\eg $\Ex \pi_1 \oplus \Ex \pi_2$), it cannot express \emph{temporal nesting}.
We intentionally accept this limitation; restricting trace quantifiers to the system's initial states (or fixed prefixes) is a deliberate architectural choice required to preserve QHL's applicability to black-box environments. 
In practice, the vast majority of probabilistic hyperproperties (\eg those used in security verification such as probabilistic noninterference) quantify over the system's initial states or shared origin prefixes and, as we will show in Section~\ref{sec:evaluation}, are expressible in QHL.

\section{Statistical Verification of QHL}\label{sec:statistical_verification}

\subsection{Overview} To verify QHL, we propose a simulation-based statistical verification (SV) approach. 
Rather than computing the exact value of the system's underlying probability measure, we estimate the quantitative valuation $\valuation{ \Phi }^\mu_\Pi$ by independently sampling finite traces. 
Instead of computing the exact integrals, suprema, and infima for the $\Ex$, $\Max$, and $\Min$ quantifiers, we evaluate them via statistical approaches based on Hoeffding's inequality \cite{hoeffding1963probability} and extreme value theory (EVT) \cite{coles2001EVT}. 
Our approach provides a formal $(\varepsilon, \delta)$-guarantee, ensuring that for a given error tolerance $\varepsilon > 0$ and a confidence level $1 - \delta \in [0, 1]$, the algorithm returns an estimate $\hat{v}$ such that:
\begin{align*}
	\Pr(|\hat{v} - \valuation{ \Phi }^\mu_\Pi| \le \varepsilon) \ge 1 - \delta
\end{align*}

Achieving this formal guarantee is not merely a direct adoption of existing work in this area.
The expressiveness of QHL introduces several unique verification challenges:
(1) evaluating QHL's conditional probability measure ($\mu_{\psi, \pi}^\Pi$) requires generating samples from a restricted domain, complicating the sampling process;
(2) evaluating the extreme quantifiers ($\Max$ and $\Min$) requires applying extreme value theory, which is very different from the mean-centric approaches used in the existing works in this area \cite{agha2018survey};
(3) the presence of nested quantifiers creates a recursive evaluation tree where local confidence and error bounds accumulate and propagate upward, and depend on the underlying needed theory; \new{thus, neither statistical method alone can evaluate such a nested formula.}
Therefore, the global error ($\varepsilon$) and confidence ($\delta$) budgets must be systematically partitioned and distributed across all sub-formulas such that we can ensure their composition does not violate the global $(\varepsilon, \delta)$-guarantee.
\new{
This decomposition also suggests a natural way to classify quantitative hyperproperties by their statistical verification cost. 
Formulas whose set measure quantifiers are exclusively expectation-based ($\Ex$) admit a fixed, a priori sample bound obtained from Hoeffding's inequality, whereas any occurrence of $\Max$ or $\Min$ requires the sequential, sample-unbounded approaches, and alternations between the two compound this cost through the recursive budget partitioning developed below. 
In this sense the quantifier structure of a QHL formula plays a role analogous to that of quantifier alternation in {\hyperltl} model checking \cite{hyperltl}: not in determining decision complexity, but in determining the achievable sample complexity of statistical verification.
}
\new{%
We particularly show the impact of quantifier alternation on the sample complexity.
Note that alternation in our setting relates to alternation in the underlying statistical theories associated with a quantifier and not necessarily an alternation between dual quantifiers, as in {\hyperltl}.
}

\tightpar{Sampling}
In this work, we target black-box systems where the verification engine cannot observe or manipulate the internal state of the system, and can only stimulate it to observe the resulting trace. 
Given this black-box nature, we exclusively focus on finite traces bounded to a maximum length $l$, evaluating our quantitative metrics over the probability measure of the cylinder sets defined by these finite prefixes.
As a result, our statistical estimations and their corresponding $(\varepsilon, \delta)$-guarantees only reflect the system's behavior up to a finite execution depth (\ie $l$).
\new{
Similarly, we require relations ($R$) and functions ($f$) appearing in a specification to be $l$-bounded (\ie $R(t_1,...,t_n)$ depends only on $(t_1[..l],...,t_n[..l])$, and likewise for $f$). 
Initial-state equivalence, prefix equivalence, and metrics such as power distance or execution time up to step $l$ are $l$-bounded by construction.
}

Evaluating QHL's conditional probability measure ($\pi_{\sim {\psi}}$) within this environment requires sampling strictly from the restricted measure defined by the relational formula $\psi$. 
Because the system is a black box, we must rely on \emph{rejection sampling} to generate these conditional samples. 
The engine stimulates the system, produces a trace, and evaluates it against the constraint $\psi$. 
If the new trace satisfies $\psi$, it is accepted and passed to the statistical estimator; otherwise, it is discarded, and the engine draws a new sample.

Rejection sampling preserves our black-box assumptions; it can, however, be computationally expensive, especially when the relational constraint $\psi$ is very complex. 
In this setting, the trace constraint relations $R$ that take the form of initial-state equivalence (\eg $\pi_1[0] = \pi_2[0]$) or finite history equivalence (\eg $\pi_1[..k] = \pi_2[..k]$ where $k \ll l$) are the most efficient, as they allow the engine to reject the invalid traces by simply comparing a short history of the traces, reducing the overhead of rejection sampling.
For systems that support state-saving, the engine can extract the required prefix from the trace assignment map $\Pi$ and initialize the simulator directly to that state, bypassing rejection entirely. 
Alternatively, the engine can simulate the candidate trace step-by-step and employ \emph{early termination}: the moment the prefix diverges from the required constraint, the simulation is immediately halted and rejected.
Incorporating ideas such as importance or adaptive sampling could scale up the sampling process, but we leave investigating the applicability of such techniques to future work.

\tightpar{Statistical Valuation}
Once valid conditional samples are generated and their corresponding measures calculated, the statistical approach used to aggregate them depends on the type of the quantifier. 
For $\Max$ and $\Min$, we rely on extreme value theory, and for $\Ex$, we rely on Hoeffding's inequality. 
Hereafter, we explain the approach used for each quantifier separately, and then detail how these methods are combined to evaluate nested quantifiers and entire QHL formulas.

\subsection{Expectation Quantifier ($\Ex$)}
To evaluate $\valuation{ \Ex \pi_{\sim {\psi}} . \Phi }^\mu_\Pi$, we employ Hoeffding's inequality and assume the underlying valuation of $\Phi$ is bounded within a known interval $[a, b]$, where $a , b \in \mathbb{R}$. 
By Hoeffding's inequality, to guarantee error and confidence margins $(\varepsilon, \delta)$, the required number of sampled traces $N$ is determined by:
\begin{align*}
	N \ge \frac{(b-a)^2 \ln(2/\delta)}{2\varepsilon^2}
\end{align*}
The algorithm draws $N$ independent and identically distributed (i.i.d.) traces $t_1, \dots, t_N$ from the probability measure space as defined by~$\psi$ and $\mu$. 
The expected value is then estimated via the sample mean:
\begin{align*}
	\hat{v}_{\Ex} = \frac{1}{N} \sum_{i=1}^{N} \valuation{ \Phi }_{\Pi[\pi \mapsto t_i]}
\end{align*}

\subsection{Maximum Quantifier ($\Max$)}
To evaluate $\valuation{ \Max \pi_{\sim {\psi}} . \Phi }^\mu_\Pi$, we employ EVT, which statistically models the distribution of extreme values.
The statistical verification engine draws an initial set of $M$ samples and uses the Peaks-Over-Threshold (POT) method \cite{coles2001EVT} to extract the local maximums. 
To do so, it picks a high threshold $u$ (\eg the 95th percentile of the initial $M$ samples).
For each sampled trace evaluation $x$, it keeps the values exceeding this threshold ($x > u$) and records their excesses ($y = x - u$). 
Under the assumption that the system's underlying distribution resides within the maximum domain of attraction (MDA) \cite{coles2001EVT}, the Pickands-Balkema-de Haan theorem \cite{pickands1975statistical} shows that as the threshold $u$ approaches the true absolute maximum, the distribution of these excesses \emph{asymptotically} converges to a Generalized Pareto Distribution (GPD).

A GPD over the excesses is characterized by three parameters: 
(1) a location parameter which defines the starting point of the distribution and in our case, since we are using POT, this value is always $0$ (starting point of the excess dataset), (2) a scale parameter $\sigma > 0$, which defines the statistical dispersion of the extreme values, and (3) a shape parameter $\xi \in \mathbb{R}$, which governs the tail behavior of the distribution. 
The value of $\xi$ fundamentally determines whether the maximum of the system is bounded. 
If $\xi \ge 0$, the distribution exhibits an infinite, unbounded upper tail. 
Conversely, if $\xi < 0$, the distribution exhibits a finite, strict upper bound.

Using Maximum Likelihood Estimation (MLE) \cite{casella2024statistical}, the algorithm fits the sampled excesses to a GPD, assigning values to $\xi$ and~$\sigma$. 
The asymptotic worst-case bound is then computed as the right endpoint of the fitted distribution, shifted back by the threshold $u$:
\begin{align*}
	\hat{v}_{\Max} = 
	\begin{cases} 
		u - \frac{\sigma}{\xi} & \text{if } \xi < 0 \qquad \text{(Bounded upper tail)} \\
		\infty & \text{if } \xi \ge 0 \qquad \text{(Unbounded upper tail)}
	\end{cases}
\end{align*}
If the true maximum is a priori known to be bounded within a specific interval $[a, b]$, the MLE search space can be constrained to $\xi < 0$ such that the resulting right endpoint satisfies $\hat{v}_{\Max} \le b$, thereby improving the performance.

Because the MLE parameters are statistical estimates, a sequential sampling approach is used to dynamically draw \emph{new} samples until the confidence interval of the estimated GPD endpoint shrinks to the specified error tolerance $\varepsilon$. 
By the asymptotic normality of maximum likelihood estimators \cite{casella2024statistical}, the parameter estimates for $\sigma$ and $\xi$ converge to a normal distribution around their true values for a sufficiently large sample size $M$. 
The algorithm relies on the Fisher Information matrix and the Delta Method \cite{coles2001EVT} to approximate the variance of the estimated endpoint $\hat{v}_{\Max}$.
The square root of this variance provides the Standard Error (SE) of the predicted bound. 
Given the desired confidence level $1-\delta$, the corresponding standard score ($Z$-score) for the normal distribution can be easily computed.
The error margin for the current number of samples $M$ is then calculated as $\hat{\varepsilon} = Z \times \text{SE}$.

Because SE is inversely proportional to the sample size, $\hat{\varepsilon}$ monotonically shrinks as $M$ grows. 
To provide the required $(\varepsilon, \delta)$ guarantee, the algorithm checks $\hat{\varepsilon} \le \varepsilon$. 
If the condition is not met, the engine draws additional samples, recalculates the excesses, and refits the GPD until the required tolerance is achieved.

\tightpar{Heuristic Prediction of Sample Size}
While the exact number of the next batch of samples required for the EVT estimation to converge cannot be known a priori, we can leverage the statistical properties of the estimation to approximate a sample size and make the process more efficient. 

By the asymptotic normality of maximum likelihood estimators, the variance of the endpoint estimate decreases at a rate proportional to $1/N$ as the sample size $N$ increases. 
Consequently, the standard error, and by extension, the current empirical error bound $\hat{\varepsilon}$ computed via the Z-score scale proportionally with the inverse square root of the sample size:
\begin{align*}
	\hat{\varepsilon} \propto \frac{1}{\sqrt{N}}
\end{align*}

We can exploit this asymptotic relationship to estimate the total target sample size, $N_{\text{target}}$, required to achieve our required error tolerance $\varepsilon_{\text{target}}$. 
Suppose the algorithm has drawn $N_{\text{current}}$ samples and computed the current empirical error $\hat{\varepsilon}_{\text{current}}$. 
By taking the ratio of these proportional errors, the underlying distribution constants (such as the true variance and the Z-score) cancel out entirely:
\begin{align*}
	\frac{\hat{\varepsilon}_{\text{current}}}{\varepsilon_{\text{target}}} \approx \frac{ \frac{1}{\sqrt{N_{\text{current}}}} }{ \frac{1}{\sqrt{N_{\text{target}}}} } = \sqrt{ \frac{N_{\text{target}}}{N_{\text{current}}} }
\end{align*}

Solving this for the target sample size yields a direct heuristic bound:
\begin{equation} \label{eqn:target_heuristic}
	N_{\text{target}} \approx N_{\text{current}} \times \left( \frac{\hat{\varepsilon}_{\text{current}}}{\varepsilon_{\text{target}}} \right)^2 
\end{equation}

From an implementation standpoint, this prediction formula can help optimize the EVT process. 
Rather than blindly drawing fixed-size batches until convergence, the algorithm draws an initial set of $M$ samples to calculate $\varepsilon_{\text{current}}$, plugs that into formula (\ref{eqn:target_heuristic}), and estimates the remaining number of samples $N_{\text{target}}$ required to achieve $\varepsilon_{\text{target}}$.

Note that this number is just an estimate, and we still need the sequential sampling approach. 
This is because formula (\ref{eqn:target_heuristic}) assumes the underlying shape of the tail is perfectly stable; however, since EVT is asymptotic, as we get more samples the MLE algorithm gets more accurate, and the true values of $\sigma$ and $\xi$ might change slightly.

\subsection{Minimum Quantifier ($\Min$)}
To evaluate $\valuation{ \Min \pi_{\sim {\psi}} . \Phi }^\mu_\Pi$, we apply a lower-tail variant of EVT, which mirrors the procedure used for the $\Max$ quantifier. 
Instead of searching for upper extremes, the verification engine selects a low threshold $u$ and collects the sampled evaluations falling strictly below it.
Rather than calculating excesses, the algorithm records \emph{deficits} ($y = u - x$) and fits them to a GPD.
The underlying statistical mechanics remain identical, and the asymptotic minimum bound is computed as:
\begin{align*}
	\hat{v}_{\Min} = 
	\begin{cases} 
		u + \frac{\sigma}{\xi} & \text{if } \xi < 0 \qquad \text{(Bounded lower tail)} \\
		-\infty & \text{if } \xi \ge 0 \qquad \text{(Unbounded lower tail)}
	\end{cases}
\end{align*}

All subsequent steps (sequential sampling, MLE variance approximation, and dynamic error bounding ($\hat{\varepsilon} \le \varepsilon$)) remain exactly as they do for the $\Max$ quantifier, ensuring the statistical $(\varepsilon, \delta)$ guarantee for the lower bound.

While our EVT-based approach effectively captures the semantics of $\Max$ and $\Min$ quantifiers, it fails when the trace measure evaluates to a strict Boolean domain ($\valuation{\phi}^\mu_\Pi \in \{0, 1\}$), as EVT requires the underlying data to exhibit a continuous tail. 
To evaluate such qualitative hyperproperties, we propose an asymmetric Sequential Probability Ratio Test (SPRT) \cite{agha2018survey}, which we detail in Appendix~\ref{sec:appendix:sprt}.

\subsection{Statistical Verification of QHL Formulas}
Statistical verification of QHL formulas poses a challenge due to the presence of nested quantifiers. 
Evaluating a formula $\Phi$ with $k$ nested quantifiers gives rise to a recursive sampling tree of depth~$k$. 
To evaluate the outermost path $\pi_1$, the valuation engine must recursively sample inner paths $\pi_2, \dots, \pi_k$ to fully populate the trace assignment map $\Pi$ and evaluate the inner, quantifier-free formula~$\phi$. 
Therefore, the problem of $(\varepsilon, \delta)$-statistical verification must be decomposed into the problem of recursively finding the answer of each level and then aggregating their results.
Thus, to guarantee that the final estimate $\hat{v}$ satisfies the global $(\varepsilon, \delta)$ bound at the root, we must systematically distribute the error tolerance $\varepsilon$ and the confidence budget $\delta$ across all nodes in the evaluation tree. 
We separate this distribution into two distinct mechanisms: structural error accumulation and statistical confidence partitioning.
\new{
The proofs of the Lemmas and Theorems of this section are presented in Appendix~\ref{sec:appendix:proofs}. 
}

\subsubsection{Error Accumulation and Tolerance Distribution}
To ensure the global error tolerance $\varepsilon$ is not violated, we must account for how local estimation errors propagate upward through the syntax tree. 
Error accumulates in two distinct ways: horizontally across arithmetic operations, and vertically through nested quantifiers.

\begin{restatable}[Arithmetic Error Accumulation]{lemma}{ArithmeticErrorLemma}\label{lemma:arithmetic_error_accumulation}
	Let $\hat{v}_1$ and $\hat{v}_2$ be statistical estimates of sub-formulas $\Phi_1$ and $\Phi_2$ under the assignment $\Pi$, satisfying error bounds $\varepsilon_1$ and $\varepsilon_2$.  
	For any arithmetic operation $\oplus \in \{+, -\}$, the combined estimate $\hat{v} = \hat{v}_1 \oplus \hat{v}_2$ satisfies the guarantee:
	\begin{align*}
		|\hat{v} - \valuation{\Phi_1 \oplus \Phi_2}^\mu_\Pi| \le \varepsilon_1 + \varepsilon_2
	\end{align*}
\end{restatable}

\begin{restatable}[Nested Error Accumulation]{lemma}{NestedErrorAccumulation}\label{lemma:nested_error_accumulation}
	Let a quantifier node evaluate a trace measure over an inner sub-formula $\Phi$ that carries a maximum estimation error of $\varepsilon_{\text{inner}}$. 
	If the sampling process of the quantifier itself introduces an estimation error bounded by $\varepsilon_{\text{node}}$, the total error of the evaluation is bounded by $\varepsilon_{\text{node}} + \varepsilon_{\text{inner}}$.
\end{restatable}

\subsubsection{Confidence Partitioning}
While error tolerance ($\varepsilon$) accumulates additively, the probability of statistical failure ($\delta$) accumulates via the union bound. 
At any given node with an allocated confidence budget $\delta$, we must partition this budget to cover all independent risks of failure.

\paragraph{Arithmetic nodes} For arithmetic nodes ($\oplus$), the risk is split uniformly among the two operand branches. 
A node-level failure occurs if either of those branches fails to satisfy its local bound; thus we assign each branch $\delta_{\text{left}} = \delta_{\text{right}} = \delta / 2$.

\begin{restatable}[Arithmetic Confidence Partitioning]{lemma}{ArithmeticConfidencePartitioningLemma}\label{lemma:arithmetic_confidence_partitioning}
	Let an arithmetic node ($\oplus$) be allocated total confidence budget $\delta$. 
	If the probability of statistical failure for each of its two child operands is bounded by $\delta_{\text{left}} = \delta_{\text{right}} = \delta / 2$, then the total failure probability of the arithmetic node is bounded by $\delta$.
\end{restatable}

\paragraph{Quantifier nodes} For quantifier nodes ($\Ex, \Max, \Min$), the algorithm samples $N$ traces and recursively queries the remainder of the formula ($\Phi_{\text{child}}$) for an answer to estimate the node's true measure. 
The estimation at this level suffers from two independent sources of error:
\begin{itemize}
	\item \textit{Concentration Error ($E_{\text{conc}}$)}: 
	The aggregate estimation at the current level is based on finitely many samples ($N$), which may deviate from the true population measure.
	
	\item \textit{Child Error ($E_{\text{child}}$)}:
	The recursive evaluations of the $N$ sampled traces are themselves statistical estimates, containing inherited empirical bias bounded by their own \emph{local} $\varepsilon$ and $\delta$ parameters.
\end{itemize}

To satisfy the node's total confidence budget $\delta$, we split it into $\delta_{\text{conc}}$ to bound the concentration error, and $\delta_{\text{child\_total}}$ to bound the cumulative error of all recursive children, such that $\delta_{\text{conc}} = \delta_{\text{child\_total}} = \delta / 2$.

\begin{restatable}[Quantifier Confidence Decomposition]{lemma}{QuantifierConfidenceDecompositionLemma}\label{lemma:quantifier_confidence_decomposition}
	Let a quantifier node be allocated a total confidence budget of $\delta$. 
	Bounding the probability of concentration failure by $\delta_{\text{conc}}$ and the probability of any recursive child failing by $\delta_{\text{child\_total}}$, and setting $\delta_{\text{conc}} = \delta_{\text{child\_total}} = \delta / 2$ guarantees the total failure probability of the node is strictly bounded by $\delta$.
\end{restatable}

To successfully bound the child error such that $\Pr(E_{\text{child}}) \le \delta_{\text{child\_total}}$, the budget $\delta_{\text{child\_total}}$ must be distributed among the $N$ sampled traces, assigning each $j$-th child a local confidence budget $\delta_{\text{child\_}j}$. 
The method of distribution depends on whether the sample size $N$ is fixed or unbounded:
\begin{itemize}
	\item \textit{Fixed ($\Ex$):} 
	$N$ is known \emph{a priori} via Hoeffding's inequality. 
	Thus, we apply a standard Bonferroni correction, allocating each child a budget:
	\begin{align*}
		\delta_{\text{child\_}j} = \delta_{\text{child\_total}} / N
	\end{align*}
	
	\item \textit{Unbounded ($\Max, \Min$):} 
	Extreme Value Theory requires dynamic, sequential sampling until the EVT estimation converges within the acceptable error tolerance. 
	Because $N$ is unbounded, a uniform distribution is impossible. 
	To guarantee that the dynamically allocated confidence pieces given to each child sum up to no more than $\delta_{\text{child\_total}}$, we distribute the budget using a convergent polynomial sequence \cite{apostol2013introduction}:
	\begin{align*}
		\delta_{\text{child\_}j} = \delta_{\text{child\_total}} \cdot \left( \frac{6}{\pi^2 \cdot j^2} \right)
	\end{align*}
\end{itemize}

\begin{restatable}[Quantifier Confidence Distribution]{lemma}{QuantifierConfidenceDistributionLemma}\label{lemma:quantifier_confidence_distribution}
	Let a quantifier node be allocated a concentration budget $\delta_{\text{conc}}$ and a total child budget $\delta_{\text{child\_total}}$. 
	If the node draws $N \in \mathbb{N}$ samples, and each $j$-th sample is recursively evaluated with a confidence budget $\delta_{\text{child\_}j}$ derived either uniformly (for fixed $N$) or via the polynomial sequence (for unbounded $N$), the cumulative probability of child failure strictly satisfies $\Pr(E_{\text{child}}) \le \delta_{\text{child\_total}}$.
\end{restatable}

\subsubsection{The Recursive Verification Algorithm}
Equipped with these structural bounds, the statistical verifier takes as input the global tolerance $\varepsilon$, global confidence $\delta$, and a formula $\Phi$ containing $k_Q$ quantifiers. 
The procedure operates recursively:

\begin{enumerate}
	\item \textit{Tolerance Distribution:} 
	\Cref{lemma:arithmetic_error_accumulation} and \Cref{lemma:nested_error_accumulation} establish that error tolerance $\varepsilon$ accumulates additively. 
	Consequently, the maximum possible error at the root of the evaluation tree is the direct sum of the estimation errors from every underlying quantifier node. 
	To guarantee that this total cumulative error strictly satisfies the global tolerance $\varepsilon$, we distribute the budget uniformly across the $k_Q$ quantifiers in the formula. 
	Each individual quantifier $i$ is therefore assigned a local error tolerance 
		$\varepsilon_i = \frac{\varepsilon}{k_Q}$
	
	\item \textit{Confidence Partitioning:} 
	For quantifier $i$, the incoming local confidence budget $\delta_i$ (where the root receives $\delta_1 = \delta$) is partitioned based on the node type:
	\begin{itemize}
		\item \textit{Arithmetic ($\oplus$):} 
		The incoming budget is divided uniformly among the operands, where each branch receives $\delta_{\text{left}} = \delta_{\text{right}} = \delta_i / 2$.
		\item \textit{Quantifier ($\Max, \Min, \Ex$):} 
		The budget is split equally between the concentration error ($\delta_{\text{conc}} = \delta_i / 2$) and the recursive children error ($\delta_{\text{child\_total}} = \delta_i / 2$).
	\end{itemize}
	
	\item \textit{Sampling \& Recursive Evaluation:} 
	The node draws traces from the system, evaluating each by recursively calling the algorithm on the inner formula $\Phi_{\text{child}}$. 
	The allocated confidence for each call depends on the quantifier:
	\begin{itemize}
		\item \textit{Expectation ($\Ex$):} 
		The node computes a fixed sample size $N$ \emph{a priori} using Hoeffding's inequality (based on $\varepsilon_i$ and $\delta_{\text{conc}}$). 
		Each child is queried with a uniform confidence budget $\delta_{\text{child\_}j} = \delta_{\text{child\_total}} / N$.
		
		\item \textit{Extreme Quantifiers ($\Max$, $\Min$):} 
		The node draws batches of samples dynamically until EVT converges. %
		The $j$-th child is queried with the dynamically decaying confidence budget $\delta_{\text{child\_}j} = \delta_{\text{child\_total}} \cdot \left( \frac{6}{\pi^2 \cdot j^2} \right)$.
	\end{itemize}
	
	\item \textit{Base Case:} 
	At depth $k_Q+1$, no quantifiers remain. 
	The fully populated trace assignment map $\Pi$ is used to evaluate the inner formula $\valuation{ \phi }^\mu_{\Pi}$. 
	The resulting value is returned up the recursive chain.
\end{enumerate}

\begin{restatable}[Statistical Verification Soundness]{theorem}{StatisticalVerificationSoundnessTheorem}\label{thm:sv_soundness}
	Given a QHL formula $\Phi$ containing $k_Q$ quantifiers, a global error tolerance $\varepsilon > 0$, and a global confidence parameter $\delta \in (0, 1)$, the recursive statistical verification algorithm returns an estimate $\hat{v}$ such that
		$\Pr\bigl(|\hat{v} - \valuation{ \Phi }^\mu_\Pi| \le \varepsilon\bigr) \ge 1 - \delta$
\end{restatable}

\subsection{Discussion on Unbounded Sampling}

\subsubsection{Batching Optimization for Dynamic Sampling}
From a computational perspective, recalculating bounds and dynamically updating the confidence budget for every individual recursive sample $j$ has a significant overhead. 
To optimize our statistical verifier, we introduce batch sampling. 
Instead of drawing samples one by one, the engine draws them in fixed-size blocks of size $B$. 
The polynomial sequence is then applied to batches rather than the individual samples. 
For the $m$-th batch, the aggregate batch budget is defined as:
\begin{align*}
	\delta_{\text{batch\_}m} = \delta_{\text{child\_total}} \cdot \left( \frac{6}{\pi^2 \cdot m^2} \right)
\end{align*}
Within the $m$-th batch, we apply a standard uniform Bonferroni correction, giving each of the $B$ children an identical budget of $\delta_{\text{child}} = \delta_{\text{batch\_}m} / B$. 
This improves efficiency while strictly preserving the infinite union bound established in \Cref{lemma:quantifier_confidence_distribution}.

\subsubsection{Cost of Dynamic Sampling}
While dynamic sampling via the polynomial sequence solves the unbounded sampling problem for $\Max$ and $\Min$, it introduces the cost of \emph{unconsumed} confidence budget.
The budget $\delta_{\text{child\_total}}$ is distributed across an infinite theoretical horizon to satisfy the union bound, while our engine only actually consumes the budget for the samples drawn before the EVT estimator converges. 
If the algorithm successfully reaches the $\varepsilon_i$ tolerance and halts at sample $M$, the remaining budget allocated to the tail end of the series ($\sum_{j=M+1}^{\infty} \delta_{\text{child\_}j}$) is permanently left unused. 
Consequently, the $M$ samples that were actually evaluated received a strictly smaller confidence budget than they would have if $M$ had been known a priori. 
Because smaller $\delta$ values require exponentially more recursive samples to satisfy their local bounds, this unconsumed budget translates directly into increased sample complexity for the children.

\section{Implementation and Evaluation}\label{sec:evaluation}

We have implemented the statistical verification approaches presented in Section~\ref{sec:statistical_verification} in a prototype tool called {\toolname}. 
{\toolname} is developed in Python, and takes as an input a QHL formula, the global error tolerance ($\varepsilon$) and confidence budget ($\delta$), and a \emph{Markov chain} or a \emph{simulator program} as the system model and samples it. %
Depending on the complexity of the system model and the length of the trace, the sampling process can be time-consuming. However, {\toolname} remains highly memory-efficient, as it only needs to store a single trace (or a batch of size $B$) per quantifier in memory.
{\toolname} uses the \texttt{lark} parser library \cite{larkTool} to parse QHL formulas and construct an evaluation Abstract Syntax Tree (AST). 
The verification engine recursively traverses this AST, partitions $\varepsilon$ and $\delta$, and aggregates the recursion results as described in Section~\ref{sec:statistical_verification}.
{\toolname} leverages the \texttt{scipy} library to fit the sampled results into a GPD for $\Max$ and $\Min$ quantifiers. 

In this section, we evaluate {\toolname} on three use cases related to security and privacy. 
In all these examples, we fix $\delta = 0.05$ and evaluate {\toolname} by changing the parameters relevant to the error tolerance and the  complexity of each example.

\subsection{Authentication System Vulnerability}
This use case builds on our \new{Bayes} vulnerability example from \Cref{sec:examples}. 
We consider a password-based authentication system and verify the system against the specification: 
$$\Max \pi_1 \ \Ex \pi_2 . \ \texttt{authSuccess}(\pi_1, \pi_2).$$
We model the authentication system of this use case by assuming that a user's password consists of three components: a word, a digit, and a symbol. 
We use dictionaries of various sizes for each component and assume their underlying distribution follows Zipf's law \cite{wang2017zipf} with the most common elements having the highest probability.
\new{The attacker's candidate guesses are drawn over the same dictionaries, so that they cover the space of passwords.}

To formally evaluate the vulnerability of this system in QHL, we must find the worst-case probability of successful authentication across all possible attacker guesses. 
The results of evaluating the vulnerability on this authentication system using our SV engine are summarized in Table \ref{tab:eval_vulnerability}.
As shown in the last two columns of Table \ref{tab:eval_vulnerability}, the estimated vulnerability of the system is within $\varepsilon = 0.02$ of the true vulnerability. 
This certifies the ability of {\toolname} to accurately estimate the vulnerability of the system.

\begin{table}[h]
	\centering
	\caption{Authentication System Vulnerability Results}
	\label{tab:eval_vulnerability}
	\rowcolors{3}{gray!10}{}
	\resizebox{\columnwidth}{!}{%
		\begin{tabular}{c @{\hspace{1.4em}} c c c @{\hspace{1.4em}} c c}
			\textbf{Dictionary} & \textbf{Number of} & \textbf{Time} & \textbf{Tolerance} & \textbf{Estimated} & \textbf{True} \\
			\textbf{Size} & \textbf{Samples} & (s) & $\varepsilon$  & \textbf{Vulnerability} & \textbf{Vulnerability} \\
			\midrule[1px]
			$10^{4}$ & $17,059,000$ & $142$ & $0.02$ & $0.188$ & $0.185$ \\
			$10^{6}$ & $17,059,000$ & $145$ & $0.02$ & $0.076$ & $0.074$ \\
			$10^{8}$ & $23,534,400$ & $203$ & $0.02$ & $0.069$ & $0.068$ \\
			$10^{10}$ & $23,534,400$ & $206$ & $0.02$ & $0.057$ & $0.055$ \\
			$10^{12}$ & $37,113,700$ & $322$ & $0.02$ & $0.052$ & $0.050$ \\
	\end{tabular}}
\end{table}

An interesting observation from Table \ref{tab:eval_vulnerability} is that the number of samples required remains constant across certain dictionary sizes. 
This is primarily a result of how {\toolname} is designed. 
Because the number of samples required by the Hoeffding inequality depends only on $\delta$ and $\varepsilon$, the inner $\Ex$ quantifier always requires the same number of samples regardless of dictionary size.
While the $\Max$ quantifier may theoretically need more samples, our algorithm uses fixed bucket sizes (\eg $100$), and for this specific dictionary sizes, that bucket size was sufficient for the $\Max$ quantifier to converge.%

\subsection{Power Side-Channel Analysis}
In this use case, we consider the power side-channel attack example from \Cref{sec:examples}. 
We simulate the physical power consumption of a cryptographic microchip. 
In modern hardware, power is consumed primarily when a transistor switches state (from $0$ to $1$, or vice versa). 
Therefore, we employ a standard Hamming Distance power model, where the base power consumed during a cryptographic operation is directly proportional to the number of bits that flip inside the $N$-bit hardware register \cite{li2012hamming}.
To simulate the encryption step, we approximate the S-Box and MixColumns steps of algorithms like AES by a simple multiplication and truncation function rather than actual memory-heavy lookup tables. 
This simulates the avalanche effect of the cryptographic algorithm while allowing our simulation to seamlessly scale to various register widths.
To defend against power analysis, hardware engineers often inject random noise (\eg active shielding) to drown out the cryptographic signal \cite{mangard2007power}. 
Therefore, the final power trace observed by the attacker is the sum of the base power consumption and the masking noise. %

To formally verify the effectiveness of this technique, we must evaluate the system's \emph{Quantitative Noninterference}. 
Specifically, we want to measure the absolute maximum expected leakage capacity: given an identical attacker-controlled input (the plaintext), what is the worst-case expected observable distance in power consumption (in millivolts) caused by the variation of the high-security input (the secret key)?

This brings us back to the quantitative noninterference formalization from \Cref{sec:examples} expressed by the QHL formula:
\begin{align*}
	\Max \pi_1 \ \Ex \pi_{2 \sim R_{\texttt{samePT}}(\pi_1, \pi_2)} . \ \texttt{powerDist}(\pi_1, \pi_2)
\end{align*}
where $R_{\texttt{samePT}}(\pi_1, \pi_2)$ is a constraint restricting the probability measure space of the inner expectation to traces where the attacker's chosen plaintext in $\pi_2$ matches the candidate plaintext explored in $\pi_1$. 
The secret key remains unconstrained. 
We assume access to a quantitative distance function, \texttt{powerDist},  that calculates the absolute difference in power consumption between the two traces as described above. 

Given the quantitative nature of this use case, the error tolerance $\varepsilon$ is a real number (\eg $2$ millivolts).
For a system that uses an $N$-bit register, the state space would be $2^{2N}$.
Therefore, we evaluate the expected worst-case leakage of this system for various register widths $N \in \{4, 8, 16\}$ and error tolerances.
The results are summarized in Table \ref{tab:eval_power_leakage}.

\begin{table}[h]
	\centering
	\caption{Power Side-Channel Analysis Results}
	\label{tab:eval_power_leakage}
	\rowcolors{3}{gray!10}{}
	\resizebox{\columnwidth}{!}{%
		\begin{tabular}{c @{\hspace{1.4em}} c c c @{\hspace{1.4em}} c c}
			\textbf{Register} & \textbf{Number of} & \textbf{Time} & \textbf{Tolerance} & \textbf{Estimated} & \textbf{True} \\
			\textbf{Width ($N$)} & \textbf{Samples} & (s) & $\varepsilon$ (mV) & \textbf{Leakage} & \textbf{Leakage} \\
			\midrule[1px]
			$4$-bit & $640,480$ & $2$ & $4.0$ & $7.22$ & $7.0$ \\
			$8$-bit & $1,549,700$ & $6$ & $4.0$ & $11.35$ & $12.0$ \\
			$16$-bit & $340,993,500$ & $1364$ & $4.0$ & $21.23$ & $22.0$
	\end{tabular}}
\end{table}

This table demonstrates a key advantage of QHL: its ability to provide a concrete, real-valued estimate of system behavior. 
In this use case, both the result and the error tolerance are actual physical measurements (in millivolts) showing the expected power leakage of the system. 
Using this metric, a system designer can easily determine whether the estimated leakage $\pm \varepsilon$ is acceptable or not.
Another notable observation is how the register size affects the complexity of the system. 
As shown in the second column of Table \ref{tab:eval_power_leakage}, the number of samples required to estimate the system's power leakage grows exponentially with the size of the register.

\subsection{Mixnet Anonymity}
In anonymous communication networks (\eg Mixnets), user messages are routed through intermediate nodes that batch input packets in a buffer (of size $N$) and shuffle them on output in order to break any link between sender and receiver.
The anonymity of a user is measured by an \emph{Anonymity Set Size}, which is the number of distinct users in a batch, indicating who could plausibly have been the sender of a specific intercepted packet. 
If this set is too small, the sender can be easily identified by an observer. 
Furthermore, real-world network traffic is not uniform, and a few power users generally generate the majority of traffic \cite{sarrar2012leveraging}. 
To prevent a single user from dominating an entire mix batch and destroying anonymity, mix nodes employ rate-limiting defenses (\eg a maximum of $M$ packets per user per batch). 

In this use case, we simulate a dynamic mixnet node subject to these realistic network assumptions, and evaluate its anonymity set size.
In QHL we can formulate this property as a form of quantitative deniability of the node:
\begin{align*}
	\Min \pi_1 \ \Ex \pi_{2 \sim {R_{\texttt{sameBatch}}(\pi_1, \pi_2)}} . \ \texttt{uniqueSenders}(\pi_2)
\end{align*}
where intuitively, the formula searches for the outgoing batch that offers the smallest anonymity set.
The relational trace constraint $R_{\texttt{sameBatch}}(\pi_1, \pi_2)$ locks the probability measure space to traces that generate the exact same observed outgoing batch (the public output). 
The secret input (the true sender) remains unconstrained and can be any of the users in the buffer. 
We then define a quantitative function, \texttt{uniqueSenders}, which extracts the absolute number of distinct users present in that specific batch.

The results of evaluating the mixnet anonymity set size for varying buffer sizes ($N$) and rate limits ($M$) in a mixnet with $150$ users using {\toolname} are summarized in Table \ref{tab:eval_mixnet}.

\begin{table}[h]
	\centering
	\caption{Mixnet Anonymity Results}
	\label{tab:eval_mixnet}
	\rowcolors{3}{gray!10}{}
	\resizebox{\columnwidth}{!}{%
	\begin{tabular}{c c @{\hspace{1.4em}} c c c @{\hspace{1.4em}} c c}
			\textbf{Buffer} & \textbf{Rate Limit} & \textbf{Number of} & \textbf{Time} & \textbf{Tolerance} & \textbf{Estimated} & \textbf{True} \\
			\textbf{Size ($N$)} & \textbf{($M$)} & \textbf{Samples} & (s) & $\varepsilon$ (users) & \textbf{Set Size} & \textbf{Set Size} \\
			\midrule[1px]
			$6$ & $2$ & $134,000$ & $0.4$ & $2$ &  $3.0$ & $3.0$ \\
			$6$ & $2$ & $230,100$ & $0.6$ & $1.5$ &  $3.0$ & $3.0$ \\
			$10$ & $4$ & $3,583,600$ & $10$ & $2$ & $3.0$ & $3.0$ \\
			$10$ & $4$ & $11,565,900$ & $28$ & $1.5$ & $3.0$ & $3.0$ \\
			$15$ & $4$ & $8,060,400$ & $21$ & $2$ & $3.95$ & $4.0$ \\
			$15$ & $4$ & $20,867,900$ & $55$ & $1.5$ & $3.97$ & $4.0$
	\end{tabular}}
\end{table}

Table~\ref{tab:eval_mixnet} demonstrates that {\toolname} can accurately estimate the anonymity set of this system. 
For nodes with smaller buffer sizes, this estimate is precise. 
However, as the buffer size grows, the estimate becomes less accurate, though it still remains within $\pm \varepsilon$. 
An interesting observation is the effect of $\varepsilon$ on the required number of samples, particularly for large buffers. 
As shown, lowering $\varepsilon$ (to make the estimation more accurate) increases the number of required samples, and this increase is significantly more pronounced for larger buffers.

\subsection{Use Cases from Existing Logics}
As explained in Section \ref{sec:relation-to-logics}, QHL is expressive enough to capture some of the existing qualitative and probabilistic hyper-logics.
We express and evaluate the use cases of \hyperpctlStar \cite{wang2021hyperpctl_star} and \hyperltl \cite{hyperalg} in QHL.
\new{
We illustrate that {\toolname} reaches the same verdict as {\hyperpctlStar}, but generally requires more samples, since the Hoeffding inequality is less sample-efficient than SPRT.
In the case of purely qualitative settings like in \textsc{MCHyper}, we observe that for certain specifications and models, with rare occurrences of violation, \textsc{MCHyper} will naturally report unsatisfied, whereas {\toolname} reports satisfied. This \emph{false positive} stems from {\toolname}'s statistical nature and is in line with its guarantee that a statistical verdict is valid only up to $\varepsilon$, meaning events occurring with probability below $\varepsilon$ may go undetected. 
}
Due to space limitations, these use cases are presented in Appendices \ref{sec:appendix:hyperpctl_star_usecases} and \ref{sec:appendix:hyperpltl_usecases}, respectively.

\section{Conclusion}\label{sec:conclusion}

In this paper, we initiated the study of measure-based hyper-logics, tackling a gap in the expressiveness of traditional hyper-logics in expressing quantitative hyperproperties. 
We introduced Quantitative Hyper-Logic, which replaces qualitative trace quantifiers with measure-based ones and extends temporal predicates with richer quantitative expressions, allowing us to accurately capture the quantitative nature of real-world systems.
To verify these specifications, we developed and studied a set of statistical algorithms that combine Hoeffding's inequality and extreme value theory. 
We evaluated the expressiveness of QHL and the applicability of our algorithms through several use cases focused on quantitative information flow control.

Future work includes exploring new heuristics to improve the efficiency of the verification algorithm, investigating importance and adaptive sampling for rare events, and examining further measure-based quantifiers.

\bibliographystyle{IEEEtran}
\bibliography{ref}

\iffull
\appendices
\section{Proofs}\label{sec:appendix:proofs}

\ArithmeticErrorLemma*
\begin{proof}
	The bound follows directly from the triangle inequality: 
	\begin{align*}
		|(\hat{v}_1 \oplus \hat{v}_2) - (\valuation{\Phi_1}^\mu_\Pi \oplus \valuation{\Phi_2}^\mu_\Pi)| &\le \\
		|\hat{v}_1 - \valuation{\Phi_1}^\mu_\Pi| + |\hat{v}_2 - \valuation{\Phi_2}^\mu_\Pi| &\le \varepsilon_1 + \varepsilon_2
	\end{align*}
\end{proof}

\NestedErrorAccumulation*
\begin{proof}
	Let $\valuation{\Phi_{\text{true}}}$ be the true system measure, $\valuation{\Phi_{\text{sampled}}}$ be the \emph{hypothetical} exact measure over the $N$ finitely drawn traces (\ie assuming the inner sub-formula for each trace is evaluated with zero error), and $\hat{v}$ denote the algorithm's actual computed estimate, which relies on the statistical estimates returned by the child nodes.
	By the triangle inequality, the total estimation error is strictly bounded by:
	\begin{align*}
		|\hat{v} - \valuation{\Phi_{\text{true}}}| \le \underbrace{|\valuation{\Phi_{\text{sampled}}} - \valuation{\Phi_{\text{true}}}|}_{\text{Node Sampling Error}} + \underbrace{|\hat{v} - \valuation{\Phi_{\text{sampled}}}|}_{\text{Inherited Inner Error}}
	\end{align*}
	
	The first term isolates the error introduced purely by finite sampling. By definition of the node's statistical bounds (derived via Hoeffding's inequality or EVT), this error is constrained such that $|\valuation{\Phi_{\text{sampled}}} - \valuation{\Phi_{\text{true}}}| \le \varepsilon_{\text{node}}$. 
	The second term isolates the structural error introduced by the recursive estimations. 
	Because the aggregation operators applied over the samples ($\Max$, $\Min$, and $\Ex$) are non-expansive, the maximum deviation between the aggregate of the true child values ($\valuation{\Phi_{\text{sampled}}}$) and the aggregate of the estimated child values ($\hat{v}$) is bounded by the maximum individual deviation among the children. 
	Given that the inner formula's error is bounded by $\varepsilon_{\text{inner}}$, it follows that $|\hat{v} - \valuation{\Phi_{\text{sampled}}}| \le \varepsilon_{\text{inner}}$. 
	Substituting these bounds yields the total maximum error: 
	\begin{align*}
		|\hat{v} - \valuation{\Phi_{\text{true}}}| \le \varepsilon_{\text{node}} + \varepsilon_{\text{inner}}
	\end{align*}
\end{proof}

\ArithmeticConfidencePartitioningLemma*
\begin{proof}
	Let $E_{\text{left}}$ and $E_{\text{right}}$ denote the failure events of the left and right operand branches, respectively. 
	An error bound violation at the arithmetic node occurs if either branch fails to satisfy its statistical guarantee. 
	By the union bound, the total probability of node failure is $\Pr(\text{node failure}) = \Pr(E_{\text{left}} \cup E_{\text{right}}) \le \Pr(E_{\text{left}}) + \Pr(E_{\text{right}})$. 
	Substituting the allocated branch bounds yields $\Pr(\text{node failure}) \le \delta / 2 + \delta / 2 = \delta$.
\end{proof}

\QuantifierConfidenceDecompositionLemma*
\begin{proof}
	A statistical failure at the node level occurs if either the finite sample aggregation fails ($E_{\text{conc}}$) or any of the recursive child evaluations fails ($E_{\text{child}}$). 
	By the union bound, the total probability of failure is $\Pr(\text{node failure}) = \Pr(E_{\text{conc}} \cup E_{\text{child}}) \le \Pr(E_{\text{conc}}) + \Pr(E_{\text{child}})$. 
	Substituting the bounded risks yields $\Pr(\text{node failure}) \le \delta_{\text{conc}} + \delta_{\text{child\_total}} = \delta / 2 + \delta / 2 = \delta$.
\end{proof}

\QuantifierConfidenceDistributionLemma*
\begin{proof}
	By the union bound, the cumulative risk of child failure is bounded by the sum of their individual risks: $\Pr(E_{\text{child}}) \le \sum_{j=1}^{N} \delta_{\text{child\_}j}$. 
	\begin{itemize}
		\item For fixed $N$, $\sum_{j=1}^{N} (\delta_{\text{child\_total}} / N) = \delta_{\text{child\_total}}$ holds trivially. 
		
		\item For the unbounded case, by the Euler–Riemann zeta function \cite{apostol2013introduction}, the infinite series $\sum_{j=1}^{\infty} 1/j^2$ converges to $\pi^2/6$. 
		Multiplying each term by the inverse ($6/\pi^2$) ensures the infinite sum of fractions exactly equals $1$. 
		Thus, even if $N \to \infty$, the cumulative child risk strictly never exceeds $\delta_{\text{child\_total}}$. 
		Consequently, the union bound strictly holds to infinity:
		\begin{align*}
			\sum_{j=1}^{\infty} \delta_{\text{child\_}j} = \delta_{\text{child\_total}}
		\end{align*}
		regardless of how many samples EVT ultimately requires to converge.
	\end{itemize}
\end{proof}

\StatisticalVerificationSoundnessTheorem*
\begin{proof}
	The proof proceeds by structural induction over the syntax tree of $\Phi$. 
	
	First, we bound the total estimation error. The algorithm distributes the global tolerance uniformly, assigning each quantifier a local tolerance of $\varepsilon / k_Q$. 
	By \Cref{lemma:arithmetic_error_accumulation} and \Cref{lemma:nested_error_accumulation}, statistical error accumulates strictly additively across both arithmetic and nested operations. 
	Therefore, the maximum possible error is bounded by the sum of the local errors of all $k_Q$ quantifiers (\ie $\sum_{i=1}^{k_Q} \frac{\varepsilon}{k_Q} = \varepsilon$).
	
	Second, we bound the total probability of statistical failure. 
	The algorithm partitions the global confidence budget $\delta$ dynamically at each node. 
	\Cref{lemma:arithmetic_confidence_partitioning} guarantees that arithmetic splits preserve the local confidence budget. 
	For quantifier nodes, \Cref{lemma:quantifier_confidence_decomposition} ensures the concentration and child error safely split the budget, while \Cref{lemma:quantifier_confidence_distribution} guarantees that the cumulative child error strictly satisfies its allocated confidence, even under unbounded sequential sampling of EVT. 
	Because every structural node strictly preserves the union bound for its locally partitioned budget, the sum of all failure probabilities across the entire evaluation tree never exceeds the root's initial budget $\delta$.
	
	Since the maximum aggregated error is strictly bounded by $\varepsilon$ and the total probability of any sub-estimation failing is bounded by $\delta$, the global estimate $\hat{v}$ successfully satisfies the $(\varepsilon, \delta)$ guarantee.
\end{proof}

\section{Extreme Quantifiers over Boolean Domains}\label{sec:appendix:sprt}
While the quantitative functions $f$ in QHL allow for continuous metrics (\eg execution time or energy consumption), QHL formulas can also be used to express standard qualitative hyperproperties (\eg {\hyperltl}) where the inner trace measure evaluates to a strict Boolean domain, $\valuation{\phi}^\mu_\Pi \in \{0, 1\}$. 
Our EVT-based approach fails under such circumstances, as the EVT requires the underlying data to exhibit a continuous tail.
The MLE calculates the shape $\xi$ and scale $\sigma$ of the GPD by analyzing the variance and rate of decay among the extreme values. 
Therefore, the algorithm needs to observe how these extreme values decay as they approach the absolute supremum to mathematically predict where the distribution terminates.
For strictly Boolean trace measures, the POT excesses will be discrete. 
Therefore, any threshold $u \in (0, 1)$ will yield excesses of exactly $1 - u$ for every satisfying trace, and no excesses for non-satisfying traces. 
Because every recorded extreme value is identical, the variance of the tail is exactly zero. 
Consequently, MLE estimation for the shape parameter $\xi$ will fail to converge, yielding undefined bounds. 

\subsection{Sequential Probability Ratio Test}
To evaluate probabilistic bounds over Boolean domains, statistical model checkers traditionally employ the Sequential Probability Ratio Test (SPRT) \cite{agha2018survey}. 
SPRT frames the evaluation as a statistical hypothesis testing problem between two competing hypotheses concerning $p$, the true probability of sampling a trace that satisfies the property:
\begin{itemize}
	\item \textit{Null Hypothesis ($H_0 : p \le p_0$):} The hypothesis in which the probability of finding a trace that satisfies the property is practically negligible.
	\item \textit{Alternate Hypothesis ($H_1 : p \ge p_1$):} The hypothesis in which a satisfying trace occurs with a probabilistically significant frequency.
\end{itemize}

SPRT sequentially accumulates statistical evidence by calculating a Log-Likelihood Ratio (LLR) after each sampled trace. 
If the trace evaluates to $1$, the LLR steps up by $\ln(\frac{p_1}{p_0})$ and if it evaluates to $0$, the LLR steps down by $\ln(\frac{1 - p_1}{1 - p_0})$. 
The algorithm continuously draws samples until the LLR crosses one of two decision boundaries:
\begin{align*}
	\text{Upper Boundary } (A) &= \ln \left( \frac{1 - \beta}{\alpha} \right) \quad \implies \text{Accept } H_1 \\
	\text{Lower Boundary } (B) &= \ln \left( \frac{\beta}{1 - \alpha} \right) \quad \implies \text{Accept } H_0
\end{align*}

Crossing the $A$ boundary, the algorithm accepts $H_1$ and crossing the $B$ boundary accepts $H_0$.
Here $\alpha$ defines the acceptable False Positive rate (accepting $H_1$ when $H_0$ is true), and $\beta$ defines the acceptable False Negative rate (accepting $H_0$ when $H_1$ is true).

\subsection{Asymmetric SPRT for $\Max$ Quantifier}
While standard SPRT effectively measures if an average probability mass exceeds a specified threshold, the QHL $\Max$ operator is fundamentally an \emph{existential search}, returning $1$ even if \emph{one} single satisfying witness trace exists in the system's trace space. 
Standard SPRT is ill-suited for this type of strict existential search, because it treats satisfying traces as just another sample.
If the SPRT algorithm sees a sequence of many $0$s followed by a single $1$, it treats the $1$ as a statistical outlier, absorbs it into the average, and potentially still returns $0$.

To strictly enforce the existential semantics of the QHL $\Max$ quantifier, we propose an asymmetric, hybrid verification algorithm.

\tightpar{The supremum override}
When the valuation engine draws a sample, we introduce a conditional check to catch the satisfying (supremum) trace samples. 
If the sampled trace evaluates to $1$, the engine bypasses the statistical test entirely, halts the SPRT, and immediately returns $\hat{v} = 1$. 
This guarantees that a discovered witness is never smoothed over as a statistical anomaly, perfectly preserving the supremum semantics of the $\Max$ quantifier.

\tightpar{The statistical drift}
The SPRT algorithm is going to be used only to process non-satisfying traces ($0$s) in order to provide statistical guarantee for the absence of satisfying traces (and supremum being $0$). 

Since evaluating a pure existential quantifier over an infinite state space is statistically undecidable (as finding a witness where $p \to 0$ would require infinite samples), we relax the absolute supremum to a $(1-\varepsilon)$-quantile supremum. 
By setting the alternate hypothesis exactly to our allocated error tolerance ($p_1 = \varepsilon$), we shift $\varepsilon$ from a \emph{value error} to a \emph{probability mass limit}.
Which means we are effectively looking for the existence of witness traces with a probability mass of \emph{at least} $\varepsilon$

The null hypothesis, on the other hand, needs to capture the complete absence of such witness traces.
If the maximum of a system is truly $0$, there are strictly zero witness traces in the system's entire state space, which dictates setting $p_0 = 0$ in the null hypothesis.
While theoretically setting $p_0 = 0$ would cause a problem for the standard SPRT (as LLR step up would attempt to evaluate $\ln(p_1 / 0)$), it is safe in our setting as the supremum override returns as soon as the algorithm finds a satisfying witness trace.

Assigning $p_1 = \varepsilon$ and $p_0 = 0$ into the SPRT logic, the step size for observing unsatisfying traces (\ie $0$) simplifies to a constant:
\begin{align*}
	\text{LLR Step} = \ln \left( \frac{1 - p_1}{1 - p_0} \right) = \ln \left( \frac{1 - \varepsilon}{1 - 0} \right) = \ln(1 - \varepsilon)
\end{align*}
Since $\varepsilon \in (0, 1)$, the term $\ln(1 - \varepsilon)$ is strictly negative. 
This creates a highly optimized, monotonic \emph{statistical drift}. 
The engine continuously accumulates this negative penalty for every sampled $0$, steadily grinding the LLR downward until it crosses the lower boundary $B$, at which point the SPRT would return $0$ as the supremum.

\tightpar{One-Sided Error Bounding}
In standard SPRT, the decision boundaries are governed by both $\alpha$ (the False Positive rate) and $\beta$ (the False Negative rate). 
However, our asymmetric approach alters this error space, as any satisfying witness ($1$) triggers the supremum override, the system asserts $H_1$ with absolute mathematical certainty (\ie zero statistical error). 
Consequently, our asymmetric SPRT is never used to assert $H_1$, and because the LLR step size is always negative ($\ln(1 - \varepsilon)$), the LLR monotonically decreases. 
The Upper Boundary ($A$) is therefore mathematically unreachable and entirely irrelevant to the verification process.
The asymmetric SPRT's only responsibility is to assert $H_0$ (returning $\hat{v} = 0$) by crossing the Lower Boundary ($B$).
Consequently, there is only one possible statistical failure mode in our verifier: a False Negative, where the engine crosses $B$ and asserts $H_0$ despite a witness actually existing with a probability mass $p \ge \varepsilon$.
To satisfy the union bound and guarantee our global $(\varepsilon, \delta)$ criteria, this specific failure mode must be bounded by the confidence budget. 
Therefore, we allocate the entire statistical risk budget to the False Negative rate, setting $\beta = \delta$. 
Applying the lower boundary approximation $B = \ln\left(\frac{\beta}{1 - \alpha}\right)$ under the condition that statistical False Positives are \emph{impossible} ($\alpha = 0$), the lower boundary $B$ simplifies to:
\begin{align*}
	B = \ln(\delta)
\end{align*}
Thus, the approach will be reduced to a one-sided test, as the SPRT engine accumulates the negative penalty $\ln(1 - \varepsilon)$ for every $0$ sampled until the LLR falls below $\ln(\delta)$. 
At which point, the engine stops with a formal $1 - \delta$ confidence that the probability mass of a witness is less than $\varepsilon$.

\subsection{Error Bounds Preservation}
This asymmetric verification algorithm satisfies the global QHL guarantee $\Pr(|\hat{v} - \valuation{ \Phi }^\mu_\Pi| \le \varepsilon) \ge 1 - \delta$ across both cases:
\begin{itemize}
	\item \textit{Witness Discovered ($\hat{v} = 1$):} 
	Because the engine found a literal witness, the true mathematical supremum is definitively $1$, which means the error is $|1 - 1| = 0 \le \varepsilon$ and the confidence is $100\%$, satisfying the $\ge 1 - \delta$ requirement.
	
	\item \textit{SPRT Halts ($\hat{v} = 0$):} 
	SPRT guarantees that crossing the lower boundary $B$ provides a formal $(1 - \delta)$-confidence guarantee that the true probability $p$ favors $H_0$ over $H_1$. 
	Because we defined $H_1$ based on our error tolerance ($p_1 = \varepsilon$), SPRT guarantees that the probability of finding a witness is $p < \varepsilon$. 
	Under our relaxed \emph{probability mass limit} definition, this means that the true \emph{statistical} valuation would be $0$, yielding an error of $|0 - 0| = 0 \le \varepsilon$ with a $(1 - \delta)$-confidence.
\end{itemize}

\subsection{Asymmetric SPRT for $\Min$ Quantifier}
To enforce the semantics of the $\Min$ quantifier over Boolean domains, we employ an asymmetric SPRT that mirrors the one used for the $\Max$ operator. 

\tightpar{The infimum override}
Similar to how the $\Max$ engine bypasses the statistical test upon discovering a $1$, the $\Min$ engine introduces a conditional check to catch falsifying (infimum) trace samples. 
If the sampled trace evaluates to $0$, the engine immediately halts the SPRT and returns $\hat{v} = 0$. 
This preserves the strict infimum semantics, ensuring that a discovered zero-valuation witness is never smoothed over as a statistical anomaly.

\tightpar{The statistical drift}
The SPRT algorithm is used to process satisfying traces ($1$s) to provide a statistical guarantee for the absence of falsifying traces (and the true infimum being $1$). 
We relax the infimum search to a probability mass limit, testing whether the probability of finding a $0$ is at least $\varepsilon$. 
Thus, the engine accumulates the negative penalty $\ln(1 - \varepsilon)$ exclusively for every sampled $1$. 
It would continuously process these $1$s, driving the LLR downward until it crosses the lower boundary $B = \ln(\delta)$, at which point it would halt to provide $\hat{v} = 1$ with a $(1-\delta)$-confidence guarantee that the probability mass of finding a $0$ is less than $\varepsilon$.

\section{\textsc{MCHyper} Use Cases}\label{sec:appendix:hyperpltl_usecases}
To demonstrate the translation from \hyperltl to QHL as described in Section \ref{sec:relation_to_qualitative_logics}, we selected a few security-related {\hyperltl} specifications from the literature \cite{hyperalg}. 
These specifications have previously been implemented and evaluated using \textsc{MCHyper} \cite{mchyperTool}. 
We translated them into QHL and evaluated them using {\toolname}.

In our evaluation of \hyperltl hyperproperties we focused on the I2C bus master use case of \cite{hyperalg}.
I2C is a widely used bus protocol that connects multiple components in a master-slave topology. 
A typical setup consists of one master, one controller, and several slaves. 
The master communicates to the slaves via two physical wires, the clock line (SCL) and the data line (SDA). 
The interface of the master towards the controller consists of 8 bit wide words for input and output of data, a 3-bit wide address to encode slave numbers, a system clock input, and several reset and control signals.
These use cases were specified in Aiger circuits and {\toolname} was modified to simulate and sample these circuits.

This use case focused on the information flow properties of an I2C bus. 
Even though the I2C bus has no security features, it has been used in security-critical applications, such as smart cards.
We checked this system using {\toolname} against the following information flow properties:

\paragraph{\textbf{Property (NI1)}} states that there is no information flow with respect to the address to which the I2C master intends to send data:

\begin{align*}
	\forall \pi \ \forall \pi' . &\LTLsquare (\overline{\text{ADDR\_I}}_\pi = \overline{\text{ADDR\_I}}_{\pi^{'}}) \\ &\Rightarrow \LTLsquare (\text{SDA\_O}_\pi = \text{SDA\_O}_{\pi^{'}})
\end{align*}%

This information flow is intended, and \textsc{MCHyper} reports the violation.

QHL expresses this property as:
\begin{align*}
	\Min \pi \ \Min \pi' . \ \texttt{NI1}(\pi, \pi')
\end{align*}
where the indicator function \texttt{NI1} simply models the LTL part of the (NI1) formula.

The result of evaluating (NI1) using {\toolname} is reported in Table \ref{tab:mchyper}. 
As we can see, {\toolname} gives a false positive and reports $1$, indicating (NI1) was satisfied.
Given the guarantees of our SV engine, this means that the violating trace was rare, and its probability was less than $\varepsilon$.

\paragraph{\textbf{Property (NI3)}} states that when the write enable bit is not set, no information should flow from the controller inputs to the bus.

\begin{align*}
	\forall \pi \ \forall \pi' . &\LTLsquare (\neg \text{WEn} \wedge \overline{\text{ADDR\_I}}_\pi = \overline{\text{ADDR\_I}}_{\pi^{'}}) \\ &\Rightarrow \LTLsquare (\text{SDA\_O}_\pi = \text{SDA\_O}_{\pi^{'}})
\end{align*}%

Similarly, QHL expresses this property as:
\begin{align*}
	\Min \pi \ \Min \pi' . \ \texttt{NI3}(\pi, \pi')
\end{align*}
where the indicator function \texttt{NI3}, again, simply models the LTL part of the formula.

\textsc{MCHyper} reports that this property is satisfied by the implementation, and {\toolname} verifies this with an error tolerance of $\varepsilon = 0.03$.
The result of evaluating (NI3) is also reported in Table \ref{tab:mchyper}.

\begin{table}[h]
	\centering
	\caption{Evaluation Results for \textsc{MCHyper} Use Cases}
	\label{tab:mchyper}
	\rowcolors{2}{}{gray!10}
	\resizebox{\columnwidth}{!}{%
		\begin{tabular}{c @{\hspace{1.4em}} c c c @{\hspace{1.4em}} c}
			\textbf{Property} & \textbf{Samples} & \textbf{Time (s)} & \textbf{Tolerance} $\varepsilon$ & \textbf{Verdict} \\
			\midrule[1px]
			(NI1) & $90,000$ & $14$ & $0.05$ & $1$ \\
			(NI3) & $230,000$ & $37$ & $0.03$ & $1$ \\
	\end{tabular}}
\end{table}

One notable observation is that for certain specifications marked as \emph{unsatisfied} by \textsc{MCHyper}, {\toolname} returned $1$ (\ie satisfied). 
This \emph{false positive} is rooted primarily in the statistical nature of {\toolname}. 
While a model checker like \textsc{MCHyper} can find rare falsifying traces by exploring the entire state space, a statistical approach may miss them. 
This aligns with our guarantee that {\toolname}'s verdict is valid up to $\varepsilon$, meaning events occurring with a probability lower than $\varepsilon$ may not be detected. 
For future work, we plan to improve {\toolname} using techniques such as importance sampling \cite{agha2018survey} to better identify these rare trace occurrences.

\section{\hyperpctlStar Use Cases}\label{sec:appendix:hyperpctl_star_usecases}
In this section we present the use cases from {\hyperpctlStar} \cite{wang2021hyperpctl_star} and express and evaluate them in {\toolname}.

\paragraph{\textbf{Timing Side-Channel Vulnerability}}
Timing side-channel attacks occur when an attacker can infer secret values by observing the execution time of a program. 
In this case study, we analyze the authentication algorithm of the GabFeed chat server, which is known to have a vulnerability that leaks the number of set bits in the secret key. 
By observing the total execution time across different public keys, an attacker can infer the secret. 
To formally prove the system is secure against this timing side-channel, the probability of termination (reaching termination state $F$) within $k \in \mathbb{N}$ steps should be approximately equal for two arbitrary executions starting with different secret keys $s_1$ and $s_2$. 
This requirement is expressed in \cite{wang2021hyperpctl_star} as:
\begin{align*}
	\mathbb{P}^{\pi_1}((\LTLcircle s_1^{\pi_1}) \Rightarrow (\LTLdiamond^{\le k} F^{\pi_1})) \approx_\varepsilon \mathbb{P}^{\pi_2}((\LTLcircle s_2^{\pi_2}) \Rightarrow (\LTLdiamond^{\le k} F^{\pi_2}))
\end{align*}
which ensures that the probability of termination for the secret keys is approximately equal within some $\varepsilon > 0$ ($\approx_\varepsilon$).

In QHL, we express this property by defining two relational trace constraints, $R_{s_1}$ and $R_{s_2}$, which restrict our sampling space to traces that begin with secret keys $s_1$ and $s_2$, respectively. 
We then define an indicator function, \texttt{evnTerm}, which returns $1$ if the trace terminates and $0$ otherwise. 
QHL can express this by computing the absolute difference ($|\cdot|$) between the two conditional expectations:
\begin{align*}
	\Big| \Ex \pi_{1 \sim {R_{s_1}(\pi_1)}} . \ \texttt{evnTerm}(\pi_1) - \Ex \pi_{2 \sim {R_{s_2}(\pi_1)}} . \ \texttt{evnTerm}(\pi_2) \Big|
\end{align*}

We can ensure approximate equivalence ($\approx_\varepsilon$) by verifying that this absolute difference is bounded by $\varepsilon$ (\eg $\le \varepsilon$).
In this use case, we model $k$ by the length of the trace. 
The results of evaluating the timing side-channel vulnerability using {\toolname} for various trace lengths are summarized in Table \ref{tab:eval_side_channel}.

\begin{table}[h]
	\centering
	\caption{Timing Side-Channel Results}
	\label{tab:eval_side_channel}
	\rowcolors{3}{gray!10}{}
	\resizebox{\columnwidth}{!}{%
		\begin{tabular}{c @{\hspace{1.4em}} c c c @{\hspace{1.4em}} c c}
			\textbf{Trace} & \textbf{Number of} & \textbf{Time} & \textbf{Tolerance} & \textbf{Absolute} &  \\
			\textbf{Length $k$} & \textbf{Samples} & (s) & $\varepsilon$ & \textbf{Difference} & $ \le \varepsilon$ \\
			\midrule[1px]
			$50$ & $203,008$ & $10$ & $0.01$ & $0.173$ & $\times$ \\
			$100$ & $203,008$ & $13$ & $0.01$ & $0.248$ & $\times$ \\
			$150$ & $203,008$ & $17$ & $0.01$ & $0.011$ & $\times$
	\end{tabular}}
\end{table}

{\toolname} reaches the same verdict as \hyperpctlStar, that the GabFeed chat server indeed has a timing side-channel vulnerability.
Comparing the numbers in Table~\ref{tab:eval_side_channel} with what was reported in \cite{wang2021hyperpctl_star}, we can see that {\toolname} takes longer to evaluate the formulas and report the final difference.
This is due to the fact that Hoeffding inequality requires more samples than SPRT \cite{agha2018survey} (used in the \hyperpctlStar's tool) to reach a verdict.
The power of Hoeffding inequality is that it can give us the probabilities and their differences, while hypothesis based approaches such as SPRT only return a yes/no answer \cite{agha2018survey}.
This information is reported in the fifth column of Table \cref{tab:eval_side_channel}.

\paragraph{\textbf{Randomized Cache Replacement Policy}}
Randomized cache replacement policies are countermeasures against cache flush attacks by deciding which cache lines are replaced in case of a cache miss.
The performance requirement dictates that starting from an empty cache, after an initial $N$ steps (when the cache is almost full), the probability of observing exactly zero misses ($T$ consecutive hits) within a time window of $T$ should be strictly greater than the probability of observing exactly one miss in that same window by a margin of $\varepsilon$. 

In \hyperpctlStar this requirement is expressed as:
\begin{align*}
	\mathbb{P}^{\pi_1}(\LTLcircle^{(N)}\Box^{\le T}H^{\pi_1}) > \mathbb{P}^{\pi_2}(\LTLcircle^{(N)}\varphi^{\pi_2}) + \varepsilon
\end{align*}
where $\varphi^{\pi_2}$ means that there is exactly one miss for $N$ consecutive accesses, and $\LTLcircle^{(N)}$ represents the $N$ fold composition of $\LTLcircle$.

In QHL, we abstract $\LTLcircle^{(N)}$ into quantitative indicator functions evaluated over the specified $T$-length window. 
Let \texttt{0miss} evaluate to $1$ if the trace yields zero misses in the window, and let \texttt{1miss} evaluate to $1$ if the trace yields exactly one miss. 
Unlike previous properties, these two executions are completely independent and do not require relational trace constraints. 
The QHL formula directly computes the absolute difference between expected values across the probability measure space $\mu$:
\begin{align*}
	\Big| \big(\Ex \pi_{1} . \ \texttt{0miss}(\pi_1)\big) - \big(\Ex \pi_{2} . \ \texttt{1miss}(\pi_2)\big) \Big| 
\end{align*}
and we ensure that the result is bigger than $\varepsilon$ (\ie $> \varepsilon$).

This evaluation considers a cache with 256 lines running a program of 1024 blocks.
The results of this evaluation for different trace lengths $T \in \{10, 20\}$ are summarized in Table \ref{tab:eval_cache_replacement}.
As reported in the last column on the table, {\toolname} reaches the same verdict as \hyperpctlStar, that the randomized cache replacement policy is indeed satisfied in this system.

\begin{table}[h]
	\centering
	\caption{Randomized Cache Replacement Policy Results}
	\label{tab:eval_cache_replacement}
	\rowcolors{3}{gray!10}{}
	\resizebox{\columnwidth}{!}{%
		\begin{tabular}{c @{\hspace{1.4em}} c c c @{\hspace{1.4em}} c c}
			\textbf{Trace Length} & \textbf{Number of} & \textbf{Time} & \textbf{Tolerance} & \textbf{Absolute} &  \\
			($T$) & \textbf{Samples} & (s) & $\varepsilon$ & \textbf{Difference} & $ > \varepsilon$ \\
			\midrule[1px]
			$10$ & $8,122$ & $2$ & $0.05$ & $0.43$ & $\checkmark$ \\
			$10$ & $8,122$ & $2$ & $0.05$ & $0.44$ & $\checkmark$ \\
			$20$ & $203,008$ & $55$ & $0.01$ & $0.14$ & $\checkmark$ \\
			$20$ & $203,008$ & $57$ & $0.01$ & $0.13$ & $\checkmark$ \\
	\end{tabular}}
\end{table}

\paragraph{\textbf{Probabilistic Noninterference}}
This case study uses the same multithreaded system that was introduced in Section\ref{sec:evaluation}, in which the thread execution lengths and the termination probabilities depended on the high input $h$.
Probabilistic noninterference was formulated in \hyperpctlStar as:
\begin{align*}
	\mathbb{P}^{\pi_1}\Big((\LTLcircle h_0^{\pi_1}) &\Rightarrow (\LTLdiamond(F^{\pi_1} \wedge l_0^{\pi_1}))\Big) \\ &\approx_\varepsilon \mathbb{P}^{\pi_2}\Big((\LTLcircle h_1^{\pi_2}) \Rightarrow (\LTLdiamond(F^{\pi_2} \wedge l_0^{\pi_2}))\Big)
\end{align*}
by ensuring the $\varepsilon$-equivalence of the probability of two random execution traces ($\pi_1$ and $\pi_2$) one starting from $h=0$ and the other from $h=1$, both terminating with low output $l=0$.
A symmetric formula should also be evaluated for the $l=1$ output.

In QHL, we evaluate this property by defining two relational trace constraints, $R_{h_0}$ and $R_{h_1}$, which restrict our sampling space to traces that begin with inputs $h_0$ and $h_1$, respectively. 
We then define an indicator function, \texttt{evnTermL0}, which models $\LTLdiamond(F \wedge l_0)$ returning $1$ if the trace eventually reaches a termination state with output $l=0$, and $0$ otherwise. 
QHL measures the expected value of this indicator function under both conditional measure spaces and returns the absolute difference:
\begin{align*}
	\bigg| \Big(\Ex \pi_{1 \sim {R_{h_0}(\pi_1)}} . \ &\texttt{evnTermL0}(\pi_1) \Big) \\ - &\Big(\Ex \pi_{2 \sim {R_{h_1}(\pi_2)}} . \ \texttt{evnTermL0}(\pi_2)\Big) \bigg|
\end{align*}

We check approximate equivalence ($\approx_\varepsilon$) by ensuring the result of this evaluation is bounded by $\varepsilon$.
The results of evaluating this formulation on probabilistic noninterference are summarized in Table \ref{tab:eval_noninter_hyperpctlstar}.

\begin{table}[h]
	\centering
	\caption{Probabilistic Noninterference Results}
	\label{tab:eval_noninter_hyperpctlstar}
	\rowcolors{3}{gray!10}{}
	\resizebox{\columnwidth}{!}{%
		\begin{tabular}{c @{\hspace{1.4em}} c c c @{\hspace{1.4em}} c c}
			& \textbf{Number of} & \textbf{Time} & \textbf{Tolerance} & \textbf{Absolute} &  \\
			\textbf{Threads $N$} & \textbf{Samples} & (s) & $\varepsilon$ & \textbf{Difference} & $ \le \varepsilon$ \\
			\midrule[1px]
			$10$ & $8,122$ & $0.1$ & $0.05$ & $0.118$ & $\times$ \\
			$10$ & $203,008$ & $2.53$ & $0.01$ & $0.111$ & $\times$ \\
			$25$ & $8,122$ & $0.1$ & $0.05$ & $0.158$ & $\times$ \\
			$25$ & $203,008$ & $2.57$ & $0.01$ & $0.149$ & $\times$ \\
			$50$ & $8,122$ & $0.11$ & $0.05$ & $0.169$ & $\times$ \\
			$50$ & $203,008$ & $2.9$ & $0.01$ & $0.153$ & $\times$
	\end{tabular}}
\end{table}

{\toolname} reports that this system is interfering which matches the verdict reported by \hyperpctlStar.
However, {\toolname} needs much more samples and takes a much longer time than \hyperpctlStar, which as we discussed before, was predictable due to the differences between Hoeffding inequality and SPRT \cite{agha2018survey}.

\paragraph{\textbf{Dining Cryptographers}}
In this use case, $N$ cryptographers must determine whether the National Security Agency (NSA) or one of the cryptographers paid for dinner, without revealing the identity of the payer. 
Adjacent cryptographers flip a coin to establish a 1-bit shared secret ($s_{ij}$) and publicly state whether the two coins they can see (the left- and right-hand ones)
\emph{agree} or \emph{disagree}, stating the opposite if they were the ones who paid. 
Then an even number of agrees indicates that the NSA paid, while an odd number indicates that a cryptographer paid.

The security property then ensures that if some cryptographer paid, the probability that it was cryptographer $i$ or $j$ must be approximately equal, independently of the shared secret (the coin tosses) between them. 
Let $P_i$ and $P_j$ denote the outcomes where cryptographer $i$ and cryptographer $j$ paid, respectively. 
The logic asserts that the probability of each outcome ($i$ paying vs. $j$ paying) under each coin toss condition ($s_{ij}$ vs. $\neg s_{ij}$) remains $\varepsilon$-equivalent:
\begin{align*}
	\mathbb{P}^{\pi_1}&(\LTLdiamond(\neg S_{ij}^{\pi_1} \wedge \LTLcircle P_i^{\pi_1})) 
	\approx_\epsilon \mathbb{P}^{\pi_2}(\LTLdiamond(S_{ij}^{\pi_2} \wedge \LTLcircle P_i^{\pi_2})) \\
	& \approx_\epsilon \mathbb{P}^{\pi_3}(\LTLdiamond(\neg S_{ij}^{\pi_3} \wedge \LTLcircle P_j^{\pi_3}))
	\approx_\epsilon \mathbb{P}^{\pi_4}(\LTLdiamond(S_{ij}^{\pi_4} \wedge \LTLcircle P_j^{\pi_4}))
\end{align*}

QHL handles this four-way equivalence by explicitly measuring the absolute distances between the relevant pairs. 
We define indicator functions \texttt{paid}$_i$ and \texttt{paid}$_j$ (evaluating to $1$ if the respective cryptographer paid) and relational trace constraints $R_{S_{ij}}$ and $R_{\neg S_{ij}}$ to filter traces by the outcome of the coin toss. 

To formally verify the protocol, the statistical engine evaluates the absolute difference between any two specific conditional expectations. 
For example, to prove that the coin toss leaks no information regarding cryptographer $i$ paying, the QHL formula computes:
\begin{align*}
	\Big| \Ex \pi_{1 \sim {R_{S_{ij}}(\pi_1)}} . \ \mathbb{I}_{P_i}(\pi_1) - \Ex \pi_{2 \sim {R_{\neg S_{ij}}(\pi_2)}} . \ \mathbb{I}_{P_i}(\pi_2) \Big|
\end{align*}

As before, we check approximate equivalence ($\approx_\varepsilon$) by ensuring the result is bounded $\le \varepsilon$.
The results of evaluating the probabilistic anonymity property for varying numbers of cryptographers ($N \in \{100, 1000\}$) using {\toolname} are summarized in Table \ref{tab:eval_dining_crypto}.

\begin{table}[h]
	\centering
	\caption{Dining Cryptographers Results}
	\label{tab:eval_dining_crypto}
	\rowcolors{3}{gray!10}{}
	\resizebox{\columnwidth}{!}{%
		\begin{tabular}{c @{\hspace{1.4em}} c c c @{\hspace{1.4em}} c c}
			\textbf{Cryptographers} & \textbf{Number of} & \textbf{Time} & \textbf{Tolerance} & \textbf{Absolute} &  \\
			($N$) & \textbf{Samples} & (s) & $\varepsilon$ & \textbf{Difference} & $ \le \varepsilon$ \\
			\midrule[1px]
			$100$ & $8,122$ & $0.1$ & $0.05$ & $0.0145$ & $\checkmark$ \\
			$100$ & $8,122$ & $0.1$ & $0.05$ & $0.0006$ & $\checkmark$ \\
			$1000$ & $203,008$ & $2.5$ & $0.01$ & $0.0044$ & $\checkmark$ \\
			$1000$ & $203,008$ & $2.5$ & $0.01$ & $0.0069$ & $\checkmark$ \\
	\end{tabular}}
\end{table}

As reported in the last column of the table, {\toolname} reaches the same verdict as \hyperpctlStar, that the dining cryptographers protocol does indeed leak no information regarding cryptographer $i$ paying.

\else
\fi

\end{document}